\documentclass[aoas]{imsart}

\usepackage{amsthm,amsmath,amsfonts, amssymb}
\RequirePackage[colorlinks,citecolor=blue,urlcolor=blue]{hyperref}

\usepackage[round]{natbib}

\arxiv{arXiv:1703.06633}

\startlocaldefs
\usepackage{xcolor}
\usepackage{enumerate}
\usepackage[normalem]{ulem}
\usepackage{graphicx}
\usepackage{multirow}
\usepackage{xspace}
\usepackage[position=top]{subfig}

\newtheorem{lemma}{Lemma}
\newtheorem{proposition}{Proposition}
\newtheorem{remark}{Remark}

\newcommand{\tr}{\text{tr}}

\newcommand{\Cov}{{\mathbb C}\text{ov}\xspace}
\newcommand{\Esp}{\mathbb E}

\newcommand{\Rbb}{\mathbb R}
\newcommand{\Nbb}{\mathbb N}
\newcommand{\Ncal}{\mathcal{N}}
\newcommand{\Pcal}{\mathcal{P}}
\newcommand{\Qcal}{\mathcal{Q}}
\newcommand{\Var}{\mathbb V}

\newcommand{\eqdef}{\triangleq}

\newcommand{\matr}[1]{\boldsymbol{#1}}
\newcommand{\vect}[1]{\matr{#1}} 
\newcommand{\trans}{\intercal}
\newcommand{\transpose}[1]{\matr{#1}^\trans}
\newcommand{\crossprod}[2]{\transpose{#1} \matr{#2}}
\newcommand{\tcrossprod}[2]{\matr{#1} \transpose{#2}}

\DeclareMathOperator*{\diag}{diag}

\newcommand{\hatB}{\widehat{\matr{B}}}
\newcommand{\hatTheta}{\widehat{\matr{\Theta}}}
\newcommand{\hatSigma}{\widehat{\matr{\Sigma}}}

\newcommand{\tildeM}{\widetilde{\matr{M}}}

\newcommand{\tildeZ}{\widetilde{\matr{Z}}}

\newcommand{\tildeP}{\widetilde{\matr{P}}}
\newcommand{\pPCA}{pPCA\xspace}
\endlocaldefs

\begin{document}

\begin{frontmatter}

 \title{Variational inference for probabilistic Poisson PCA}
 \runtitle{Variational inference for probabilistic Poisson PCA}
 
 \begin{aug}
 
 \author{\fnms{Julien} \snm{Chiquet}${}^{1}$\corref{}},
 \author{\fnms{Mahendra} \snm{Mariadassou}${}^{2}$}
 \and
 \author{\fnms{St\'ephane} \snm{Robin}${}^{1}$}
 \runauthor{Chiquet, Mariadassou, Robin}
 
 \ead[label=e1]{julien.chiquet@inra.fr}
 \ead[label=e2]{mahendra.mariadassou@inra.fr}
 \ead[label=e3]{robin@agroparistech.fr}

 \address{
 MIA-Paris\\
 UMR 518 AgroParisTech / INRA \\
 AgroParisTech\\
 16, rue Claude Bernard \\
 75231 Paris CEDEX 05, France \\\printead{e1,e3}} 

 \address{INRA Unit\'e MaIAGE\\
 Bât. 233 et 210 \\
 Domaine de Vilvert\\
 78352 JOUY-EN-JOSAS CEDEX, France \\ \printead{e2}}

\affiliation{\small ${}^{1}$UMR MIA-Paris, AgroParisTech, INRA, Universit\'e Paris-Saclay,  Paris, France\\
\small ${}^{2}$MaIAGE, INRA, Universit´e Paris-Saclay, Jouy-en-Josas, France}
 \end{aug}
 
 \begin{abstract}
   Many application domains such as ecology or genomics have to deal
   with multivariate non Gaussian observations. A typical example is
   the joint observation of the respective abundances of a set of
   species in a series of sites, aiming to understand the
   co-variations between these species. The Gaussian setting provides
   a canonical way to model such dependencies, but does not apply in
   general. We consider here the multivariate exponential family
   framework for which we introduce a generic model with multivariate
   Gaussian latent variables. We show that approximate maximum
   likelihood inference can be achieved via a variational algorithm
   for which gradient descent easily applies. We show that this
   setting enables us to account for covariates and offsets. We then
   focus on the case of the Poisson-lognormal model in the context of
   community ecology.  We demonstrate the efficiency of our algorithm
   on microbial ecology datasets. We illustrate the importance of
   accounting for the effects of covariates to better understand
   interactions between species.
 \end{abstract}


\end{frontmatter}

\section{Introduction} \label{sec:Intro}


Principal component analysis (PCA) is among the oldest and most
popular tool for multivariate analysis. It basically aims at reducing
the dimension of a large data set made of continuous variables
\citep{And03, MKB79} in order to ease its interpretation and
visualization. The methodology exploits the dependency
structure between the variables to exhibit the few synthetic variables
that best summarize the information content of the whole data set: the
principal components. In that sense, PCA can be viewed as a way to
better understand the dependency structure between the variables. From
a purely algebraic point-of-view, PCA can be seen as a
matrix-factorization problem where the data matrix is decomposed as
the product of a loading matrix and a score matrix \citep{EcY36}.

For statistical purposes, PCA can also be cast in a probabilistic
framework. Probabilistic PCA (\pPCA) is a model-based version of PCA
originally defined in a Gaussian setting, in which the scores are
treated as random hidden variables \citep{TiB99,Min00}. It is closely
related to factor analysis. As it involves hidden variables,
maximum-likelihood estimates (MLE) can be obtained via an EM algorithm
\citep{DLR77}. One major interest of the probabilistic approach is
that it allows to combine dimension reduction with other modeling
tools, such as regression on some available covariates. Because
observed variables can be affected by the variations of such
covariates, the correction for their potential effects is desirable to
avoid the presence of spurious correlations between the responses.

The Gaussian setting is obviously convenient as the dependency
structure is entirely encoded in the covariance matrix but \pPCA has
been extended to more general settings. Indeed, in many applications
\citep{RoW05,SrL10}, Gaussian models need to be adapted to handle
specific measurement types, such as binary or count data. For count
data, the multivariate Poisson distribution seems a natural
counterpart of the multivariate normal. However, no canonical form
exist for this distribution \citep{JKB97}, and several alternatives
have been proposed in the literature including Gamma-Poisson
\citep{Nel85} and lognormal-Poisson \citep{AiH89,Izs08}. The latter
takes advantage of the properties of the Gaussian distribution to
display a larger panel of dependency structure than the former, but
maximum likelihood-based inference raises some issues as the MLE of
the covariance matrix is not always positive definite.


A series of works have contributed to extend PCA to a broader class of
distributions, typically in the exponential family. The matrix
factorization point-of-view has been adopted to satisfy a positivity
constraint of the parameters \citep{Laf15} and to minimize the Poisson
loss function \citep{CaX15} or more general losses \citep{LeS01}
consistent with exponential family noise. Sparse extensions have also
been proposed \citep{WTT09,LDS16}. In a model-based context,
\cite{CDS01} suggest to minimize a Bregman divergence to get estimates
of the scores: the divergence is chosen according to the distribution
at hand and a generic alternating minimization scheme is
proposed. \cite{SHD14} consider a similar framework and use matrix
factorization for the minimization of Bregman divergence. In both
cases, the scores are considered as fixed parameters. \cite{MGH09}
cast the same model in a Bayesian context and use Monte-Carlo sampling
for the inference. \cite{AGZ15} consider Bayesian inference of the
Gamma-Poisson distribution. A Bayesian version of PCA (where both
loadings and scores are treated as random) is considered in
\cite{LiT10}.

\cite{Lan15} reframes exponential family PCA as an optimization
problem with some rank constraints and develops both a convex
relaxation and a maximization-minimization algorithm for binomial and
Poisson families. Finally \cite{ZHD12} and \cite{Zho16} consider
factor analysis in the more complex setting of negative-binomial
families. Our approach differs from the previous ones as we only
consider scores as random variables, whereas we consider the loadings
as fixed parameters, in the exact analog of \citeauthor{TiB99}'s pPCA.

As recalled above, in \pPCA, the scores are treated as hidden
variables. One of the main issue of non-Gaussian \pPCA arises from the
fact that their conditional distribution given the observed data is
often intractable, which hampers the use of an
Expectation-Minimization (EM) strategy. Variational approximations
\citep{JaJ00,WaJ08} have become a standard tool to approximate such
conditional distributions. \cite{Kar05} uses such an approximation for
the inference of the one-dimensional Poisson-lognormal model and
derives a variational EM (VEM) algorithm. \cite{HOW11} provide a
theoretical analysis of this approximation for the same model and
prove the consistency of the estimators. Indeed, even the conditional
distribution of one single hidden coordinate (given all others) is
unknown, which makes regular Gibbs sampling inaccessible. As a
consequence, \cite{LeS01} use moment estimates, whereas, in a Bayesian
context, \cite{LiT10} resort to a variational approximation of the
conditional distribution.

\paragraph{Our contribution} We define a general framework for \pPCA
in the simple exponential family. The model we consider combines
dimension reduction (via \pPCA) and regression, in order to account
for known effects and focus on the remaining dependency
structure. Scores are assumed to be Gaussian to allow a large panel of
dependency structures. We put a special emphasis on the analysis of
count data.  We adopt a frequentist setting rather than a Bayesian
approach to avoid non-scalable, computing-heavy Monte-Carlo
sampling. We use a variational approximation of the conditional
distribution of the scores given the observed data to derive a
variational lower bound of the likelihood. Since only the scores are
assumed to be random, we can prove that this bound is biconcave,
\emph{i.e.} concave in the model parameters and in the variational
parameters but not jointly concave in general. Biconcavity allows us
to design a gradient-based method rather than a (variational) EM
algorithm, traditionally used in this setting.

We illustrate the interest of our model on two examples of microbial
ecology. We show that the proposed algorithm is efficient for large
datasets such as these encountered in metagenomics. We also show the
importance of accounting for covariates and offset, in order to go
beyond first-order effects. More specifically, we show how the
proposed modeling allows us to distinguish between correlations that
are caused by known covariates from those corresponding to an unknown
structure and requiring further investigations.

The paper is organized as follows: in Section \ref{sec:Model} we introduce pPCA for the exponential family and the variational framework that we consider. Section \ref{sec:Extension} generalizes the model in the manner of a generalized linear model, in order to handle covariates and offsets. Then, Section \ref{sec:Inference} is dedicated to the inference and optimization strategy. Section \ref{sec:poisson} details the special Poisson case and Section \ref{sec:Visualization} devises the visualization, an important issue for non-Gaussian PCA methods. Finally, Section~\ref{sec:Illustration} considers applications to two examples from metagenomics: the impact of a pathogenic fungi on microbial communities from  tree leaves, and the impact of weaning on piglets gut microbiota.

\section{A variational framework for probabilistic PCA in the
 exponential family}
\label{sec:Model}

We start this section by stating the probabilistic framework
associated to Gaussian probabilistic PCA. Then we show how it can be
naturally extended to other exponential families. We finally derive
variational lower bounds for the likelihood of pPCA and its gradient,
which are the building blocks of our inference strategy.

\subsection{Gaussian probabilistic PCA (pPCA)}

The probabilistic version of principal component analysis -- or pPCA
-- \citep{Min00, MGH09, TiB99} relates a sample of $p$-dimensional
observation vectors $\vect{Y}_i$ to a sample of $q$-dimensional
vectors of latent variables $\vect{W}_i$ in the following way:
\begin{equation}
 \label{eq:ppca-model0}
\vect{Y}_i = \vect{\mu} + \matr{B} \vect{W}_i + \vect{\varepsilon}_i,
\qquad \vect{\varepsilon}_i \sim \Ncal(\vect{0}_p, \sigma^2
\matr{I}_p).
\end{equation}
The parameter $\vect{\mu}$ allows the model to have \emph{main
  effects}. The $p \times q$ matrix $\matr{B}$ captures the dependence
between latent and observed variables. Furthermore, the latent vectors
are conventionally assumed to have independent Gaussian components
with unit variance, that is to say,
$\vect{W}_i \sim \Ncal(\vect{0}_q, \matr{I}_q)$. This ensures that
there is no structure in the latent space. Model
\eqref{eq:ppca-model0} can thus be restated as
$\vect{Y}_i \sim \Ncal(\vect{\mu}, \tcrossprod{B}{B} +
\sigma^2\matr{I}_p)$.

In the following, we consider an alternative formulation stated in a
hierarchical framework. Despite its seemingly more complex nature,
it lends itself nicely to generalizations. Formally,
\begin{equation}
  \label{eq:ppca-model}
 \begin{array}{rll}
   \text{latent space} & (\vect{W}_i)_{i=1,\dots,n} \quad \text{i.i.d.},  
   & \vect{W}_i \sim \Ncal(\vect{0}_q, \matr{I}_q) \\[1ex]
   \text{parameter space} & \vect{Z}_i = \vect{\mu} + \matr{B} 
                            \vect{W}_i, & \\[1ex]
   \text{observation space} & \left(Y_{ij} | 
Z_{ij}\right)_{i=1,\dots,n;\; j=1,\dots,p} \quad \text{indep.}, 
   & Y_{ij} | Z_{ij} \sim \Ncal(Z_{ij}, \sigma^2)
 \end{array}
\end{equation}
In Equation~\eqref{eq:ppca-model}, $\vect{Z}_i$ is a linear transform
of $\vect{W}_i$ and the last layer $\vect{Y}_i|\vect{Z}_i$ simply
corresponds to \emph{observation noise}.  Informally, the
\emph{latent} variables $\vect{W}_i$ (in $\mathbb{R}^q$) are mapped to
a linear subspace of the \emph{parameter} space $\mathbb{R}^p$ via the
$\vect{Z}_i$ which are then pushed into the \emph{observation} space
using Gaussian emission laws. The main idea of this paper is to
replace Gaussian emission laws with univariate natural exponential
families.

Note that the diagonal nature of the covariance matrix of
$\vect{\varepsilon}_i$ specified in \eqref{eq:ppca-model0} now means
that, \emph{conditionally on $\vect{Z}_i$}, all components of
$\vect{Y}_i$ are independent.  This is why we may consider univariate
variables $Y_{ij} | Z_{ij}$ in Formulation
\eqref{eq:ppca-model}. Although the observation noises are
conditionally independent, the coordinates of a given $\vect{Y}_i$ are
not, which makes the model genuinely multivariate. This is further
emphasized in Section \ref{subsec:featuresPoisson}.

The loading matrix $\matr{B}$ is a convenience parameter that is
useful for both optimization and visualization of the model but not
identifiable \emph{per se}. Indeed any orthogonal transformation of $\matr{B}$ 
leads
to the same model: denoting $\matr{\Sigma} = \tcrossprod{B}{B}$, Model
\eqref{eq:ppca-model} can be rephrased as
\begin{equation*}
  \begin{split}
    (\vect{Z}_i)_{i=1,\dots,n} \, \text{i.i.d.}, \, \vect{Z}_i \sim \Ncal(\vect{\mu}, \matr{\Sigma}), \\
    \left(Y_{ij} | Z_{ij}\right)_{i=1,\dots,n;\; j=1,\dots,p} \, \text{indep.}, 
\, Y_{ij} | Z_{ij} \sim \Ncal(Z_{ij}, \sigma^2).
  \end{split}
\end{equation*}
As a consequence, the identifiable parameters of the model are $\vect{\mu}$ and $\matr{\Sigma}$.

Hereafter and unless stated otherwise, index $i$ refers to \emph{observations} 
and ranges in $\{1,\dots,n\}$, index $j$ refers to \emph{variables} and ranges 
in $\{1,\dots,p\}$ and index $k$ refers to \emph{factors} and ranges in 
$\{1,\dots,q\}$.

\subsection{Natural Exponential family (NEF)}

The work in this study is based on essential properties of univariate
\emph{natural exponential families} (NEF) where the parameter is in
canonical form. They include normal distribution with known variance,
Poisson distribution, gamma distribution with known shape parameter
(and therefore exponential distribution as a particular example) and
binomial distribution with known number of trials. The probability
density (or mass function) of a NEF can be written
\begin{equation}
 \label{eq:exp-fam-dist}
 f(x|\lambda) = \exp \left( x\lambda- b(\lambda) - a(x) \right)
\end{equation}
where $\lambda$ is the canonical parameter and $b$ and $a$ are known
functions. The function $b$ is well known to be convex (and analytic)
over its domain and the mean and variance are easily deduced from $b$
as
\begin{equation*}
 \Esp_{\lambda}[X] = b'(\lambda) \quad \text{and} \quad \Var_{\lambda}[X] = 
b''(\lambda).
\end{equation*}
The canonical link function $g$ is defined such that $g(b'(\lambda)) = \lambda$. 
The maximum likelihood estimate $\hat{\lambda}$ of $\lambda$ from a single 
observation $x$ is given by $\hat{\lambda} = \hat{\lambda}(x) = g(x)$ and 
satisfies 
\begin{equation*}
 \Esp_{\hat{\lambda}(x)}[X] = b'(\hat{\lambda}(x)) = x.
\end{equation*}

\subsection{Probabilistic PCA for the exponential family}

We now extend pPCA from the Gaussian setting to more general NEF. The
connection between the two versions is exactly the same as the
connection between linear models and generalized linear models (GLM).
Intuitively, we assume that $i)$ there exists a (low) $q$-dimensional
(linear) subspace in the \emph{natural canonical parameter space}
where some latent variable $\vect{Z}_i$ lie; and $ii)$ observations
$\vect{Y}_i$ are generated in the \emph{observation space} according to some
NEF distribution with parameter $\vect{Z}$. The latter is linked to
$\Esp[\vect{Y}_i|\vect{Z}_i]$ through the canonical link function $g$.
In the Gaussian case, the link function is the identity and the
parameter space can be identified with the observation space but this is not
the case in general for other families. Formally, we extend Model \eqref{eq:ppca-model} to
\begin{equation}
  \label{eq:pca-model} 
  \begin{split}
    (\vect{W}_i)_{i=1,\dots,n} & \,\, \text{i.i.d.}, \,\, \vect{W}_i \sim \Ncal(\vect{0}_q,\matr{I}_q) \,, \\
    \vect{Z}_i & = \vect{\mu} + \matr{B} \vect{W}_i, \, \\
    \left(Y_{ij} | Z_{ij}\right)_{i=1,\dots,,n;\; j=1,\dots,p} & \,\, 
\text{indep.}, \\ p(Y_{ij}|Z_{ij}) & = \exp\left(Y_{ij}Z_{ij} - b(Z_{ij}) - 
a(Y_{ij})\right).
  \end{split}
\end{equation}
Note in particular that $g(\Esp[Y_{ij} | Z_{ij}]) = g(b'(Z_{ij})) = Z_{ij}$ and
that an unconstrained estimate $\tilde{Z}_{ij}$ of $Z_{ij}$ is
$\tilde{Z}_{ij} = g(Y_{ij})$. The vector $\vect{\mu}$ corresponds to main 
effects, $\matr{B}$ to \emph{rescaled} loadings in the parameter spaces and 
$\vect{W}_i$ to scores of the $i$-th observation in the low-dimensional latent 
subspace. The model specified in \eqref{eq:pca-model} is 
the same as the one specified in \eqref{eq:ppca-model} but for the last data 
emission layer. Similarly to Model \eqref{eq:ppca-model}, the first two lines of Model \eqref{eq:pca-model} can be combined into 
$\vect{Z}_i$ i.i.d. such that $\vect{Z}_i \sim \Ncal(\vect{\mu}, \matr{\Sigma})$ with $\matr{\Sigma} = \tcrossprod{B}{B}$.
\begin{remark} As stated previously, $\matr{B}$ is only identifiable through $\tcrossprod{B}{B}$ and therefore at best up to rotations in $\mathbb{R}^q$. Note that this limitation is
shared with standard PCA. Intuitively, PCA finds a good $q$-dimensional affine approximation subspace $\vect{\mu} + \text{Span}(\matr{B})$ of $\matr{Y}$. But without
additional constraints infinitely many bases $\matr{B}$ can be used to 
parametrize this subspace. Orthogonality constraints and ordering of the 
principal components in decreasing order of variance are necessary to uniquely 
specify $\matr{B}$. Imposing them in standard PCA additionnally allows one to 
leverage \citeauthor{EcY36}'s theorem and reduce a $q$-dimensional 
approximation to a series of $q$ unidimensional problem. It also entails 
nestedness: the best $q$-dimensional approximation is nested within the best 
$q+1$-dimensional one and so on. There is unfortunately no equivalent in 
exponential PCA. We therefore do not force $\matr{B}$ to be orthogonal in our 
model. For visualization however, we perform orthogonalization to ensure 
consistency of the graphical outputs with standard PCA (see 
Section~\ref{sec:Visualization}).
\end{remark}

\subsection{Likelihood}

Note $\matr{Y}$ (resp. $\matr{W}$) the $n \times p$ (resp.
$n \times q$) matrix obtained by stacking the row-vectors
$\transpose{Y}_i$ (resp. $\transpose{W}_i$). Conversely, for any
matrix $\matr{A}$, $\vect{A}_i$ refers to the $i$-th row of $\matr{A}$
considered as a \emph{column} vector. In matrix expression,
$\matr{Z} = \vect{1}_n \transpose{\mu} + \tcrossprod{W}{B}$. The
observation matrix $\matr{Y}$ only depends on $\matr{Z}$ through
$\vect{\mu}$, $\matr{B}$ and $\matr{W}$ and the complete
log-likelihood is therefore
\begin{multline*}
 \log p(\matr{Y}, \matr{W} ; \vect{\mu}, \matr{B}) = \sum_{i=1}^n \log p(\vect{Y}_i | \vect{W}_i ; \vect{\mu}, \matr{B}) + \log p(\vect{W}_i) \\
 = \sum_{i=1}^n \left[ \sum_{j=1}^p Y_{ij}(\mu_j + \transpose{B}_j \vect{W}_i) - b(\mu_j + \transpose{B}_j \vect{W}_i) - a(Y_{ij}) - \sum_{k=1}^q \frac{W_{ik}^2 + \log(2\pi)}{2} \right]
\end{multline*}
which can be stated in the following compact matrix form:
\begin{multline}
 \label{eq:PCA-cond-likelihood}
 \log p(\matr{Y}, \matr{W} ; \vect{\mu}, \matr{B}) =
 \transpose{1}_n \left[\matr{Y} \odot (\vect{1}_n \transpose{\mu} +
 \tcrossprod{W}{B}) - b(\vect{1}_n \transpose{\mu} +
 \tcrossprod{W}{B})\right] \vect{1}_p \\ - \frac{\|\matr{W}\|_F^2}{2} - 
\frac{nq}{2}\log(2\pi) - K(\matr{Y}),
\end{multline}
where the function $a$ and $b$ are applied component-wise to vectors
and matrices, $\odot$ is the Hadamard product and
$K(\matr{Y}) = \transpose{1}_n a(\matr{Y}) \vect{1}_p$ is a constant
depending only on $\matr{Y}$ and not scaling with $q$.

We do not know how to integrate out $\vect{W}$ and therefore cannot
derive an analytic expression of
$\log p(\matr{Y}; \vect{\mu}, \matr{B})$.  Numerical approximation
using Hermite-Gauss quadrature or MCMC techniques are possible but
rely on computing $np$ expectations of the form
$\Esp[e^{\crossprod{a}{u} - b(\alpha + \crossprod{c}{u})}]$ for
$\vect{u} \sim \Ncal(0, \matr{I}_q)$, with $b$ nonlinear, $\vect{a}$
and $\vect{c}$ arbitrary vectors and $\alpha$ a scalar depending on
$\vect{\mu}$ and $\matr{B}$. This is likely to become computationally
prohibitive as the dimension $q$ of the latent integration space
increases. A standard EM algorithm relying on
$\Esp_{W|Y}[\log p(\matr{Y}, \matr{W} ; \vect{\mu}, \matr{B})]$ is
similarly not possible as it requires at least first and second order
of $p(\vect{W_i}|\vect{Y_i})$ which are unknown in general and as hard
to compute as the previous expectations. We resort instead to a
variational strategy and integrate out $\vect{W}$ under a tractable
approximation of $p(\matr{W} | \matr{Y})$.

\subsection{Variational bound of the likelihood}

Consider  any  product  distribution
$\tilde{p} = \otimes _{i=1}^n \tilde{p}_i$ on the $\vect{Z}_i$. The
variational approximation relies on maximizing the following lower bound over a
tractable set for $\tilde{p}$
\begin{equation*}
 \log p(\matr{Y} ; \vect{\mu}, \matr{B}) \geq J_q(\tilde{p},\vect{\mu}, \matr{B})
\end{equation*}
where
\begin{equation}
 \label{eq:pca-lower}
 \begin{aligned}
 J_q(\tilde{p}, \vect{\mu}, \matr{B}) & \eqdef \log p(\matr{Y}; \vect{\mu}, \matr{B}) - KL(\tilde{p}(\matr{W})||p(\matr{W}|\matr{Y}; \vect{\mu}, \matr{B})) \\[1ex]
 & = \Esp_{\tilde{p}} [\log p (\matr{Y}, \matr{W} ; \vect{\mu}, \matr{B}) - \log \tilde{p}(\matr{W})] \\[1ex]
 & = \sum_{i=1}^n \Esp_{\tilde{p}_i} [\log p (\vect{W}_i) + \log p(\vect{Y}_i | \vect{W}_i; \vect{\mu}, \matr{B}) - \log \tilde{p}_i(\vect{W}_i)],\\
 \end{aligned}
\end{equation}
with term-by-term inequality:
\begin{align*}
\log p(\matr{Y}_i; \vect{\mu}, \matr{B} ) & \geq J_q(\tilde{p}_i,
      \vect{\mu}, \matr{B}) \\
& \eqdef \Esp_{\tilde{p}_i} [\log p (\vect{W}_i) + \log p(\vect{Y}_i | \vect{W}_i; \vect{\mu}, \matr{B}) - \log \tilde{p}_i(\vect{W}_i)].
\end{align*}

In our variational approximation, we choose here the set $\Qcal$ of product distribution
of $q$-dimensional multivariate Gaussian with diagonal
 covariance matrices:
\begin{equation}
 \label{eq:EM-elliptic-approximation-set}
 \begin{split}
 \Qcal = \left\{\tilde{p} \eqdef \tilde{p}_{\matr{M}, \matr{S}}; \
 \tilde{p} (\matr{w}) = \prod_{i=1}^n \tilde{p}_{i}(\vect{w}_i) \right\},\\
 \text{where } \tilde{p}_{i} = \Ncal(\vect{m}_i, \diag(\vect{s}_i \odot
 \vect{s}_i)), \, (\vect{m}_i, \vect{s}_i) \in \mathbb{R}^q \times
 \mathbb{R}_+^q.
 \end{split}
\end{equation}
The $n\times q$ matrices $\matr{M}$ and $\matr{S}$ are obtained by
respectively stacking $\transpose{m}_i$ and $\transpose{s}_i$. Note that, by 
construction, $p(\matr{W} | \matr{Y})$ is a product distribution and that the 
approximation only stems from the functional form of each $\tilde{p}_i$, 
\emph{i.e.} multivariate normal with diagonal variance-covariance matrix. For 
such $\tilde{p} = \tilde{p}_{\matr{M}, \matr{S}}$, results on first and second 
order moments of multivariate Gaussian show that
\begin{multline*} 
J_q(\vect{\mu}, \matr{B}, \vect{m}_i, \vect{s}_i) \eqdef J_q(\tilde{p}_i, \vect{\mu}, \matr{B}) \\ 
= \transpose{Y}_i (\vect{\mu} + \matr{B}\vect{m}_i) - \frac12
[\|\vect{m_i}\|_2^2 + \| \vect{s}_i \|_2^2] + \frac12 (\transpose{2}_q
\log (\vect{s}_i) + q) \\ - \transpose{1}_p \Esp_{\tilde{p}_i} [ b(\vect{\mu} + \matr{B}\vect{W}_i)] - K(\matr{Y}).
\end{multline*}
Therefore,
\begin{multline}
 \label{eq:pca-vem-elliptic-likelihood} 
J_q(\vect{\mu}, \matr{B}, \matr{M}, \matr{S}) \eqdef J_q(\tilde{p}_{\matr{M}, \matr{S}}, \vect{\mu}, \matr{B}) = \sum_{i=1}^n J_q(\vect{\mu}, \matr{B}, \vect{m}_i, \vect{s}_i) \\ 
 = \transpose{1}_n \left[\matr{Y} \odot (\vect{1}_n \transpose{\mu} + \tcrossprod{M}{B}) - \Esp_{\tilde{p}} [b\left( \transpose{1}_n \vect{\mu} + \tcrossprod{W}{B} \right) ] \right] \vect{1}_p \\ 
 - \frac12 \transpose{1}_n \left[ \matr{M} \odot \matr{M} + \matr{S} \odot \matr{S} - 2\log(\matr{S}) - \matr{1}_{n, q} \right] \vect{1}_q - K(\matr{Y}).
\end{multline}

Depending on the natural exponential family and thus the exact value
of $b$ in \eqref{eq:pca-vem-elliptic-likelihood}, we may have a fully
explicit variational bound for the complete likelihood which paves the
way for efficient optimization. In particular, this is the case with
the Poisson distribution that we investigate in further details in
Section \ref{sec:poisson}.
\\

Before moving on to actual inference, we show how the framework introduced above can be extended to account for covariates and
offsets.

 
\section{Accounting for covariates and offsets} \label{sec:Extension}

Multivariate  analyses  typically   aim  at  deciphering  dependencies
between variables. Variations induced by  the effect of covariates may
be confounded with  these dependencies. Therefore, it  is desirable to
account for such  effects to focus on the  residual dependencies.  The
rational of  our approach  is to  postulate the  existence of  a model
similar  to  linear  regression  in the  \emph{parameter}  space.   We
consider the general framework  of linear regression with multivariate
outputs, which encompasses multivariate analysis of variance.

\subsection{Model and likelihood}

Suppose  that   each  observation  $i$   is  associated  to   a  known
$d$-dimensional  covariate vector  $\vect{X}_i$.  We  assume that  the
covariates  act  linearly  in  the \emph{parameter}  space  through  a
$p   \times   d$   regression  matrix   $\matr{\Theta}$,   \emph{i.e.}
$\vect{X}_i$  is linearly  related to  $\vect{Z}_i$.  It  can be  also
useful to  add an  offset to model  different sampling  efforts and/or
exposures. There  is usually one  known offset parameter  $O_{ij}$ per
observation $Y_{ij}$  and this offset  can be readily  incorporated in
our framework.  Thus, a natural generalization of \eqref{eq:pca-model} 
accounting for covariates and offsets is
\begin{equation}
  \label{eq:pca-model-with-covariates}  
  \begin{split}
    (\vect{W}_i)_{i=1,\dots,n} & \,\, \text{i.i.d.} \,\, \vect{W}_i \sim \Ncal(\vect{0}_q,\matr{I}_q) \, \\
    \vect{Z}_i & = \vect{O}_i + \matr{\Theta} \vect{X}_i + \matr{B} \vect{W}_i, \\
    \left(Y_{ij} | Z_{ij}\right)_{i=1,\dots,n; j=1,\dots,p} & \,\, \text{indep.}, \\ p(Y_{ij} | Z_{ij}) & = \exp\left(Y_{ij}Z_{ij} - b(Z_{ij}) - a(Y_{ij})\right),
  \end{split}
\end{equation}
where a column of ones can be added to the data matrix  $\mathbf{X}$ to get an intercept in the model. The log-likelihood can      be       computed      from
\eqref{eq:pca-model-with-covariates} like before to get

\begin{multline}
  \label{eq:pca-likelihood-with-covariates-and-offset}
  \log p(\matr{Y}, \matr{W} ; \matr{B}, \matr{\Theta}, \matr{O}) \\ 
  = \transpose{1}_n \left[\matr{Y} \odot (\matr{O} + \tcrossprod{X}{\Theta} + \tcrossprod{W}{B}) - b(\matr{O} + \tcrossprod{X}{\Theta} + \tcrossprod{W}{B})\right] \vect{1}_p \\ 
  -    \frac{\|\matr{W}\|_F^2}{2}    -   \frac{nq}{2}\log(2\pi)    -
  K(\matr{Y}), 
\end{multline}
where  the focus  of inference  is on  $\matr{B}$ and  $\matr{\Theta}$
while $\matr{O}$ is known.

\subsection{Variational bound of the likelihood}

We  can  use  the  variational class  $\Qcal$  previously  defined  in
\eqref{eq:EM-elliptic-approximation-set} to lower bound the likelihood
from   Eq.~\eqref{eq:pca-likelihood-with-covariates-and-offset}.    We
first introduce  the instrumental matrix $\matr{A}$,  which appears in
many equations.
\begin{equation}
   \label{eq:var-conditional-expectation-general}
  \begin{aligned}
  \matr{A} & = \Esp_{\tilde{p}}[ b \left(\matr{O} + \tcrossprod{X}{\Theta} + \tcrossprod{W}{B} \right) ] \\
           & = \Esp[ b \left( \matr{O} + \tcrossprod{X}{\Theta} + (\matr{M} + \matr{S} \odot \matr{U})\transpose{B} \right) ] 
           & = \Esp[ b \left( \matr{V} \right) ], 
  \end{aligned}
\end{equation}
where
$\matr{V} = \left( \matr{O} + \tcrossprod{X}{\Theta} + (\matr{M} +
  \matr{S} \odot \matr{U})\transpose{B} \right)$
and $\matr{U}$ is a $n \times q$ matrix with unit variance independent
Gaussian components. $\matr{V}$ can be interpreted as the variational
counterpart of $\matr{Z}$.

Since $\matr{O}$ is known, we drop  it from the arguments of $J_q$ and
obtain  the  following  lower  bound, which  extends  the  bound  from
Eq.~\eqref{eq:pca-vem-elliptic-likelihood}:
\begin{multline}
  \label{eq:pca-vem-elliptic-likelihood-general}  
  J_q(\matr{\Theta}, \matr{B}, \matr{M}, \matr{S}) 
  = \transpose{1}_n \Big(\matr{Y} \odot (\matr{O} + \tcrossprod{X}{\Theta} + \tcrossprod{M}{B}) - \matr{A} \Big) \vect{1}_p \\ 
  - \frac12 \transpose{1}_n \left[ \matr{M} \odot \matr{M} + \matr{S} \odot \matr{S} - 2\log(\matr{S}) - \matr{1}_{n, q} \right] \vect{1}_q - K(\matr{Y}).
\end{multline}


\section{Inference} 
\label{sec:Inference}

As usual in the variational framework, we aim to maximize the lower
bound $J_q$ which we call the objective function in an optimization
perspective. The optimization shall be performed on
$\matr{\Theta}, \matr{B}, \matr{M}, \matr{S}$. We only give results in
the most general case \eqref{eq:pca-vem-elliptic-likelihood-general}
with covariates and offsets. All other cases are deduced by setting
$\matr{O} = \mathbf{0}_{n \times p}$ and/or $\matr{X} = \vect{1}_n$
hereafter.

\subsection{Inference strategy}

We first highlight the biconcavity of the objective function $J_q$.
The major part of the proof is postponed to Appendix
\ref{app:convexity}.
\begin{proposition}
The variational objective function $J_q(\matr{\Theta}, \matr{B}, \matr{M}, \matr{S})$ is concave in $(\matr{\Theta}, \matr{B})$ for $(\matr{M}, \matr{S})$ fixed and vice-versa. 
\end{proposition}

\begin{proof}
 Fix  $(\matr{M},  \matr{S})$  in
 \eqref{eq:pca-vem-elliptic-likelihood-general}. The non explicit
 part of $J_q$, that is to say
 $- \transpose{1}_n \matr{A} \vect{1}_p$, is concave in
 $(\matr{\Theta}, \matr{B})$ thanks to Lemma~\ref{Lem:ConvExpComp}
 (see Appendix \ref{app:convexity}).
 By inspection, the explicit part of $J_q$ involves linear, quadratic
 and concave functions of $(\matr{\Theta}, \matr{B})$ and is also
 concave. The objective $J_q$ is therefore concave in
 $(\matr{\Theta}, \matr{B})$. The same is true for
 $(\matr{M}, \matr{S})$ when fixing $(\matr{\Theta}, \matr{B})$.
\end{proof}

A standard approach for maximizing biconcave functions is block
coordinate descents, of which the Expectation-Maximization (EM)
algorithm is a popular representative in the latent variable setting.
It is especially powerful when we have access to closed formula for
both the optimal $(\matr{M}, \matr{S})$ given
$(\matr{\Theta}, \matr{B})$ (E-step) and the optimal
$(\matr{\Theta}, \matr{B})$ given $(\matr{M}, \matr{S})$ (M-step).
However, the non-linear nature of
$\Esp_{\widetilde{p}} [b\left( \matr{O} + \tcrossprod{X}{\Theta} +
  \tcrossprod{W}{B} \right)]$
combined with careful inspection of the objective function $J_q$ shows
that setting the derivatives of $J_q$ to zero, even after fixing the
variational or model parameters, does not lead to closed formula
for $(\matr{M}, \matr{S})$ nor for $(\matr{B}, \matr{\Theta})$.
Nevertheless, since we may derive convenient expressions for the
gradient $\vect{\nabla} J_q$ (see next
Section~\ref{sec:app-gradient}), we propose to rely on the
globally-convergent method-of-moving-asymptotes (MMA) algorithm for
gradient-based local optimization introduced by
\cite{svanberg2002class} and implemented in the NLOPT optimization
library \citep{nlopt}. In the general case
\eqref{eq:pca-vem-elliptic-likelihood-general}, the total number of
parameters to optimize
$J_q(\matr{\Theta}, \matr{B}, \matr{M}, \matr{S})$ is
$p (d+q) + 2 n q$. We use box constraints for the variational
parameters $\matr{S}$ (\emph{i.e.} the standard deviations in
\eqref{eq:EM-elliptic-approximation-set} and thus only defined on
$\mathbb{R}_+^q$). The starting point is chosen according to the exact
value of $b$.

\subsection{Blockwise gradients of $J_q$} \label{sec:app-gradient}

The  blockwise  gradient  of
$J_q(\matr{\Theta}, \matr{B}, \matr{M}, \matr{S})$ can be expressed
compactly in matrix notations. We skip the tedious but straightforward
derivations and present only the resulting partial gradients. We
introduce $\matr{A}'= \Esp[ b' \left( \matr{V} \right) ]$, the natural
counterpart to matrix $\matr{A}$ given in
\eqref{eq:var-conditional-expectation-general}.  Intuitively,
$A'_{ij}$ is the conditional expectation of $Y_{ij}$ under
$\widetilde{p}_i$. On top of that, we need two other matrices denoted
$\matr{A}'_1$ and $\matr{A}'_2$, defined as follows:
\begin{equation*}
 \matr{A}'_1 = \Esp[ b' \left( \matr{V} \right)^\trans (\matr{S} \odot \matr{U}) ], 
 \quad \quad
 \matr{A}'_2 = \Esp[ (b' \left( \matr{V} \right) \matr{B}) \odot \matr{U} ].
\end{equation*}

With those matrices the derivatives of $J_q$ can be expressed compactly as
\begin{equation}
 \label{eq:var-approx-gradient-general}
 \begin{array}{rcl@{\qquad}rcl}
 \displaystyle \frac{\partial J_q}{\partial \matr{\Theta}} & = & (\matr{Y} - \matr{A'})^\trans \matr{X}, &
 \displaystyle \frac{\partial J_q}{\partial \matr{B}} & = & (\matr{Y} - 
\matr{A'})^\trans \matr{M} - \matr{A'}_1, \\[2ex]
 \displaystyle \frac{\partial J_q}{\partial \matr{M}} & = & (\matr{Y} - \matr{A'}) \matr{B} - \matr{M}, &
 \displaystyle \frac{\partial J_q}{\partial \matr{S}} & = & \displaystyle \left[ \matr{S}^{\oslash} - \matr{A'}_2 - \matr{S} \right],
 \end{array}
\end{equation}
where the $n\times q$ matrix $\matr{S}^{\oslash}$ is the elementwise
inverse of $\matr{S}$, i.e. $S_{ij}^{\oslash} = S_{ij}^{-1}$ for all
$i=1,\dots, n$, $q=1,\dots,Q$.

In the following, the resulting parameter estimates will be denoted by 
$\widehat{\matr{\Theta}}$ and $\widehat{\matr{B}}$, and the optimal variational 
parameters will be denoted by $\widetilde{\matr{M}}$ and 
$\widetilde{\matr{S}}$. We use different notation on purpose, in order to 
distinguish model parameters from variational ones.

\subsection{About missing data} When the data are missing at random
(MAR), the sampling does not disturb the inference and it is
sufficient to maximize the likelihood on the observed part of the data
\citep{LR14}. Our model can easily handle missing data under MAR
conditions as follows: note
$\Omega \subset \{1,\dots,n\} \times \{1,\dots,p\}$ the set of
observed data and $\matr{\Omega}$ the matrix where $\Omega_{ij} = 1$
if $(i,j) \in \Omega$ and $0$ otherwise. With this matrix $\Omega$,
the likelihood can be adapted from
Equation~\eqref{eq:pca-likelihood-with-covariates-and-offset}, and one
has
\begin{multline*}
 \log p(\matr{Y}, \matr{W} ; \matr{B}, \matr{\Theta}, \matr{O}) \\ 
 = \transpose{1}_n \Big( \big(\matr{Y} \odot (\matr{O} +
 \tcrossprod{X}{\Theta} + \tcrossprod{W}{B}) - b(\matr{O} +
 \tcrossprod{X}{\Theta} + \tcrossprod{W}{B})\big) \odot
 \matr{\Omega} \Big) \vect{1}_p \\ 
 - \frac{\|\matr{W}\|_F^2}{2} - \frac{nq}{2}\log(2\pi) - \tr(\transpose{\Omega} a(\matr{Y})).
\end{multline*} 
The corresponding variational bound $J_q$ and its partial derivatives
are  then  simple  adaptations  from
Equations~\eqref{eq:pca-vem-elliptic-likelihood-general}  and
\eqref{eq:var-approx-gradient-general} where $\matr{Y}$ (resp.
$\matr{A}$, $\matr{A'}$)  is replaced with
$\matr{Y} \odot \matr{\Omega}$ (resp. $\matr{A} \odot \matr{\Omega}$,
$\matr{A'} \odot \matr{\Omega}$).

Note that it is strictly equivalent for the optimization method to use
$(\matr{Y} - \matr{A'})\odot \matr{\Omega}$ or to impute missing
$Y_{ij}$ with $A'_{ij}$ before using
Eq.~\eqref{eq:var-approx-gradient-general}. Since $\matr{A'}$ is
computed as part of the gradient computation at each step, imputation
of missing data is essentially a free by-product of the optimization
method. Finally, note that $A'_{ij} = \Esp_{\widetilde{p}_i}[Y_{ij}]$
so that the imputation makes intuitive sense: we are imputing $Y_{ij}$
with its conditional expectation under the current variational
parameters.  Addressing not MAR conditions requires to take into
account the sampling process leading to missing data in order to
correctly unbias the estimation. This is out of the scope of this
paper.


\subsection{Variance estimation}
As mentioned above, only $\matr{\Theta}$ and $\matr{\Sigma}$ are identifiable parameters and an estimate of the later needs to be derived. Recall that Model \eqref{eq:pca-model-with-covariates} can be rephrased as $\vect{Z}_i \sim \Ncal(\vect{O}_i + \matr{\Theta} \vect{X}_i, \matr{\Sigma})$: it can be checked that the corresponding variational lower bound is maximal for
$$
 \widehat{\matr{\Sigma}} = \frac1n \sum_i \Esp_{\tilde{p}}\left[ (\vect{Z}_i - \vect{O}_i - \widehat{\matr{\Theta}} \vect{X}_i) (\vect{Z}_i - \vect{O}_i - \widehat{\matr{\Theta}} \vect{X}_i)^\trans \right].
$$
Since $\Esp_{\tilde{p}}(\vect{Z}_i) = \vect{O}_i + \widehat{\matr{\Theta}}
\vect{X}_i + \widetilde{\vect{M}}_i \widehat{\matr{B}}^\trans$ and
$\Var_{\tilde{p}}(\vect{Z}_i) = \widehat{\matr{B}}
\diag(\widetilde{\vect{s}}_i \odot \widetilde{\vect{s}}_i)
\widehat{\matr{B}}^\trans$, we get
$$
 \widehat{\matr{\Sigma}} = \widehat{\matr{B}} \left(\frac1n \widetilde{\matr{M}}^\trans\widetilde{\matr{M}} + \bar{\matr{S}} \right) \widehat{\matr{B}}^\trans
$$
where
$\bar{\matr{S}} = n^{-1} \diag[\vect{1}_n^\trans (\widetilde{\matr{S}}
\odot \widetilde{\matr{S}}) ]$.
Observe that $\widehat{\matr{\Sigma}}$ has rank $q$ by construction.

\subsection{Model selection}
The dimension $q$ of the latent space itself needs to be estimated. To this aim, we adopt a penalized-likelihood approach, replacing the log-likelihood by its lower bound $J_q$. We consider two classical criteria: BIC \citep{Sch78} and ICL \citep{BCG00}. We remind that ICL uses the conditional entropy of the latent variables given the observations as an additional penalty with respect to BIC. The difference between BIC and ICL measures the uncertainty of the representation of the observations in the latent space.

Because the true conditional distribution $p(\matr{W}|\matr{Y})$ is intractable, we replace it with its variational approximation $\widetilde{p}(\matr{W})$ to evaluate this entropy. The number of parameters in our model is $p(q + d)$ and the entropy of each $W_i$ under $\widetilde{p}_i$ is $q \log(2\pi e) / 2 + \sum_j \log(s_{ij})$. Based on this we define the following approximate BIC and ICL criteria:
\begin{equation}
 \label{eq:model-selection}
 \begin{aligned}
 BIC(q) & = J_q - \frac{1}{2}p(d+q)\log(n) \\
 ICL(q) & = J_q - \frac{1}{2}p(d+q)\log(n) - \frac{nq}{2} \log(2\pi e) - \transpose{1}_n \log(\matr{S}) \vect{1}_q \\
 \end{aligned}
\end{equation}
 
\section{Poisson Family}
\label{sec:poisson}

Each term of the expectation matrix $\matr{A}$ in
\eqref{eq:var-conditional-expectation-general} can be reduced to
computing expectations of the form $\Esp[b(a + cU)]$ for a convex
analytic function $b$, a standard Gaussian $U \sim \Ncal(0, 1)$ and
arbitrary scalars $(a, c) \in \mathbb{R}^2$. It can therefore be
computed numerically efficiently using Gauss-Hermite quadrature
\citep[see, e.g.,][]{NumericalRecipies}. However in the special case
of Poisson-distributed observations, $b(x) = e^x$ and most of the
expectations can be computed analytically leading to explicit formulas
for Equations~\eqref{eq:var-conditional-expectation-general}, 
\eqref{eq:pca-vem-elliptic-likelihood-general} and \eqref{eq:var-approx-gradient-general}.

\subsection{Some features of Poisson \pPCA} \label{subsec:featuresPoisson}

The Poisson \pPCA inherits some properties of the Poisson-lognormal
distribution, which states that the response vector $\vect{Y}_i$ for
sample $i$ is generated such that
$\vect{Z}_i \sim \Ncal(\vect{\mu}_i, \matr{\Sigma})$ and the
$(Y_{ij})_j$ are independent conditionally on $\vect{Z}_i$ with
$Y_{ij} | Z_{ij} \sim \Pcal(\exp(Z_{ij}))$. The moments of the
$Y_{ij}$'s are then
\begin{gather*}
  \Esp Y_{ij} = e^{\mu_j + \sigma^2_j/2},\qquad \Var Y_{ij} = \Esp Y_{ij} + (e^{\sigma^2_j} - 1) (\Esp Y_{ij})^2,\\
  \Cov(Y_{ij}, Y_{ik}) = (e^{\sigma_{jk}} - 1) \Esp Y_{ij} \Esp Y_{ik}.\\
\end{gather*}
Consequently, the Poisson-lognormal model displays both
over-dispersion of each coordinate with respect to a Poisson
distribution and pairwise correlations of arbitrary signs. In Poisson
\pPCA, $\matr{\Sigma}$ is further assumed to have a low rank.

\subsection{Explicit form of $\matr{A}$, $J_q$, and $\vect{\nabla}J_q$}

In the Poisson-case, the variational expectation of the non-linear
part involving $b$ -- the matrix of conditional expectations
$\matr{A}$ -- is equal to $\matr{A'}$ and can be expressed as
\begin{equation*}
\matr{A} = \matr{A}' = \exp\left(\matr{O} + \tcrossprod{X}{\Theta} + \tcrossprod{M}{B} + \frac12 (\matr{S} \odot \matr{S}) (\matr{B} \odot \matr{B})^\trans \right).
\end{equation*}
The lower bound $J_q$ and matrices $\matr{A}'_1, \matr{A}'_2$ appearing in \eqref{eq:var-approx-gradient-general} can be
expressed simply from $\matr{A}$ as
$$
\matr{A}'_1 = [ \transpose{A}(\matr{S} \odot \matr{S}) ] \odot \matr{B}, 
\qquad
\matr{A}'_2 = 2 [\matr{A} (\matr{B} \odot \matr{B})] \odot \matr{S}.
$$

\subsection{Implementation details}

We implemented our inference algorithm for the Poisson family in the
\texttt{R} package \textbf{PLNmodels}, the last version of which is
available on github \url{https://github.com/jchiquet/PLNmodels}.
Maximization of the variational bound $J_q$ is done using the
implementation found in the \textbf{nlopt} library \citep{nlopt} of
the globally-convergent method-of-moving-asymptotes algorithm for
gradient-based local optimization \citep{svanberg2002class}. We
interface this algorithm to R \citep{R} via the \textbf{nloptr}
package\citep{nloptr} and careful tuning of the parameters.  All
graphics are produced using the \textbf{ggplot2} package
\citep{ggplot2}.

The choice of a good starting value is crucial in iterative procedures
as it helps the algorithm to start in the attractor field of a good
local maximum and can substantially speed-up convergence. Here we
initialize $(\matr{\Theta}, \matr{B})$ by fitting a linear model to
$\log(1+\matr{Y})$, then extracting the regression coefficients
$\matr{\Theta}_{LM}$ and the variance-covariance matrix
$\matr{\Sigma}_{LM}$ of the Pearson residuals. We set
$\matr{\Theta}_0 = \matr{\Theta}_{LM}$ and
$\matr{B}_0 = (\matr{\Sigma}_{LM}^{(q)})^{1/2}$ where
$\matr{\Sigma}_{LM}^{(q)}$ is the best rank $q$ approximation of
$\matr{\Sigma}_{LM}$, as given by keeping the first $q$-dimensions of
a SVD of $\matr{\Sigma}_{LM}$. We set the other starting values as
$\matr{M}_0 = \matr{S}_0 = \matr{0}_{n \times q}$.

\section{Visualization} \label{sec:Visualization}

\subsection{Specific issues in non-Gaussian PCA}

PCA is routinely used to visualize samples in a low dimensional space.
Vizualisation in exponential PCA shares many similarities with
visualization in standard PCA, but important differences arise from
the lack of validity of \cite{EcY36}'s theorem in this setting.
\begin{enumerate}[($i$)]
\item In general, the parameter space {$\Rbb^p$} defined in
  \eqref{eq:pca-model} is different from the observation space
  {$\Nbb^p$}, as opposed to the special case of Gaussian PCA.
 \item Since principal components are not reconstructed incrementally,
   the corresponding subspaces need not be nested.
 \item The lack of constraints on $\matr{B}$ means that \emph{raw}
   scores may be correlated in the latent space, unlike their
   counterparts in standard PCA.
\end{enumerate}

To address point ($i$) we provide representations in the parameter
space as it has the Euclidean geometry practitioners are most familiar
with. Point $(ii)$ is an inherent consequence of non-linearity that
has some consequences in terms of interpretation. Indeed, the 'axis of
maximum variance' of model with rank $q$ is not the same as the first
axis of model with rank $q+1$. As for point $(iii)$, we use an
orthonormal coordinate system to represent samples in the Euclidean
parameters rather than the ``raw'' results of the algorithms. The
samples positions $\vect{Z}$ can be estimated with
$\tildeZ:=\matr{O} + \matr{X} \transpose{\hatTheta} + \tildeM
\hatB^\trans $.
$\tildeZ$ is useful to assess goodness of fit and quality of the
dimension reduction whereas $\tildeP = \tildeM \hatB^\trans$ is used
to visualize and explore structure not already captured by the
covariates.

\subsection{Quality of the dimension reduction}

A first important criterion in PCA is the amount of information that is preserved by the $q$-dimensional reduction. To this aim, we define a pseudo $R^2$ criterion, which compares the model at play to both a \emph{null} model with no latent variables and a \emph{saturated} model with one parameter per observation.

Formally, we define the matrix $\matr{\Lambda}^{(q)} = [\lambda_{ij}^{(q)}]$ where entry $\lambda_{ij}^{(q)} : = \widetilde{Z}_{ij}$ serves as an estimate of the canonical parameter of the distribution of $Y_{ij}$ given in \eqref{eq:exp-fam-dist}. We can thus define the log-likelihood $\ell_q$ of the observed data with
\begin{equation*}
\ell_q = \sum_{i=1}^n \sum_{j=1}^p [Y_{ij}\lambda_{ij}^{(q)} - \exp(\lambda_{ij}^{(q)}) ] - K(\matr{Y}).
\end{equation*}
We can compare it to the log-likelihood of the saturated model $\ell_{\max}$
(replacing $\lambda_{ij}^{(q)}$ with $\lambda_{ij}^{\max} := \log(Y_{ij})$)
and the log-likelihood $\ell_{\min}$ of the null model chosen here as a Poisson regression GLM with no latent structure, (replacing $\lambda_{ij}$ with $\lambda_{ij}^{\min} := o_{ij} + \hatTheta \vect{X}_i$, where $\hatTheta$ is estimated using a standard GLM). 
The resulting pseudo $R^2$ is defined as
\begin{equation}
 \label{eq:r-squared}
 R^2_q = ({\ell_q - \ell_{\min}}) \left/ ({\ell_{\max} - \ell_{\min}}) \right..
\end{equation}
This $R^2$ is a bit imperfect as it assumes Poisson counts, unlike the 
Poisson-lognormal in our model but it is necessary to compute equivalents to 
\emph{percentage of variance explained} that practitioners have grown accustomed 
to.

\subsection{Visualizing the latent structure}

The matrix $\tildeP	 = \tildeM \hatB^\trans$ encodes positions of the 
samples in the latent space using $\hatB$ as basis and $\tildeM$ as principal 
components. Since $\hatB$ is not constrained whatsoever, the \emph{raw} 
components are neither orthogonal nor sorted in decreasing order of variation. 
We therefore decompose $\tildeP$ as $\tildeP =  
\tildeM_{\text{viz}}\hatB^\trans_{\text{viz}}$ with columns of $\hatB_{\text{viz}}$ orthogonal and 
columns of 
$\tildeM_{\text{viz}}$ sorted in decreasing order of variation, and use 
$\tildeM_{\text{viz}}$ as principal components for visualization purposes. 
Since $\tildeP$ is already of low-rank $q$, this is achieved simply by doing a 
standard PCA of $\tildeP$. Note also that using 
either $(\hatB_{\text{viz}}, \tildeM_{\text{viz}})$ or $(\hatB, \tildeM)$ 
leaves $\hatSigma$ unchanged.

We then decompose the total variance along each component $j$ as in standard PCA. The overall contribution of axis $j$ is then $d_j \times R^2_q$, where $d_j$ is the fraction of variance in the latent space explained by component $j$.
Following the same line, we may plot the correlations between the columns of $\tildeP$ and the components arising from its PCA, to help with the interpretation of these components in terms of original variables.
\section{Illustrations} \label{sec:Illustration}

In order to highlight the scalability of our variational approach for
generalized pPCA and its flexibility for the inclusion of covariates,
we analyze two microbiome sequence count datasets below.  They
consists in counts of microbial species (OTUs or Operational Taxonomic
Units) in a series of samples. Note that an intrinsic limitation of
this sampling technology and of all marker-gene based metabarcoding
analysis methods is that they do not give access to absolute cell
counts in a sample. Indeed, these methods consist in sampling the DNA
in the biological sample and only a fixed number of DNA fragments,
referred to as the 'sequencing depth', is observed. Consequently, any
multivariate analysis of such data aims to describe the dependencies
between relative abundances \citep{Tsilimigras2016, Gloor2017},
although some additional measures can be made to recover approximate
absolute abundances \citep{Smets2015, Vandeputte2017}.

Because of the different technical steps involved in library
preparation, the sequencing depth is generally independent of the
total cell count in the sample, and its variations across samples have
no biological meaning. Therefore, the sequencing depth itself
constitutes a nuisance parameter that we need to account for to avoid
spurious correlations. To correct for varying depths across samples,
we assume that average counts scale linearly with sequencing depth,
although more sophisticated normalizations exist \citep{Chen2017}. In
subsequent analysis, the sequencing depth is just another covariate
with a special status as we know its regression coefficient and we
therefore include it as an offset in the model.  Offset is computed as
the total sequencing depth before any filter is applied to OTUs. By
doing so, the observed counts within a sample are not linearly
constrained to sum to sequencing depth.

\subsection{Impact of weaning on piglet microbiome}

\paragraph{Description of the experiment}

We considered the metagenomic dataset introduced in \cite{Mach2015}.
The dataset was obtained by sequencing the bacterial communities
collected from the feces of 31 piglets at 5 points after birth
($n = 155$). The communities were sequenced using the hypervariable
V3-V4 region of the 16S rNRA gene as metabarcoding marker gene and
sequences were processed and clustered at the 97\% identity level to
form $p = 4031$ OTUs (see \citet{Mach2015} for details of
bioinformatics preprocessing). The dataset is thus a $155 \times 4031$
count table where entry $(i, j)$ measures the relative abundance of
OTU $j$ in sample $i$ as the number of sequences (originating from
sample $i$) falling in sequence cluster $j$. One aim of this experiment
is to understand the impact of weaning on gut microbiota. Weaning, and
more generally diet changes, are well-documented to strongly impact
the gut microbiota and we therefore use weaning status as ground truth
to check whether our method can recover known structure. We also use
the example to test scalability and study how the method behaves when
the number $p$ of variables increases.

\paragraph{Numerical Experiments}

To test the impact of the number of variables on the dimension of the
latent subspace, we inferred $q$ on nested subsets of the count
table. We selected only the 3000, 2000, 1000, 500 and 100 most
abundant OTUs and fitted a model with appropriate offset to each
subset.  The offsets were chosen as log-total read count of each
sample, computed on the full OTU table. It reflects the fact that,
\emph{et ceteris paribus}, observed counts should be roughly twice as
high in communities sequenced twice more. For context, the 2500 least
abundant OTUs exhibit very high sparsity (less than 1\% of non-null
counts): each has total abundance lower than 5 and more than half
(1287) are seen only once.  It is customary to remove
such OTUs using abundance-based filters in microbiome studies. We
expect them to behave like high-dimensional noise and strongly degrade
structure recovery.

Figure~\ref{Fig:Nuria0} shows that running times increase sublinearly
with $q$ and linearly with $p$, as expected. Figure~\ref{Fig:Nuria1}
additionnally shows that low count OTUs act as high dimensional noise
and hamper our ability to recover fine structure in the latent space,
(the pseudo $R^2$ goes down from 95\% and 68\% and $\hat{q}$ from $27$
to $8$) just like it would in high dimensional Gaussian PCA.

\begin{figure}[h!]
  \centering
  \begin{tabular}{@{}c@{}}
    \hspace{-3em} running time (seconds) \\
    \includegraphics[width=.6\textwidth]{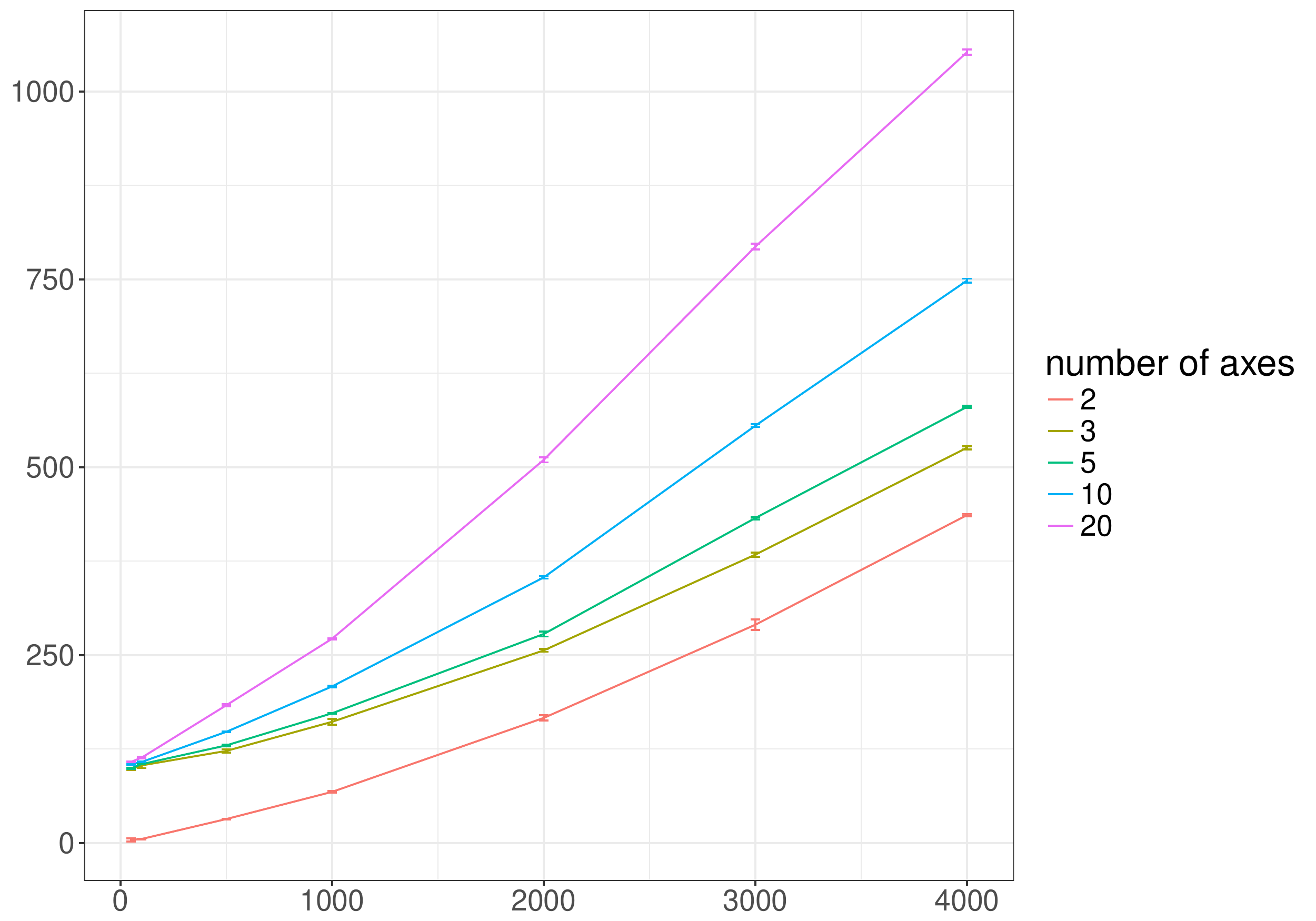} \\
    \hspace{-3em} number of variables \\
  \end{tabular}
  \caption{Dataset from \cite{Mach2015}. Running times averaged over 4
    replicates of the \texttt{PLNPCA} function in \texttt{R}
    \textbf{PLNmodels} package. Single core Intel i7-4600U CPU
    2.33GHz, \texttt{R} 3.4.1, Linux Ubuntu 16.04.}
  \label{Fig:Nuria0}
\end{figure}

\begin{figure}[h!]
  \centering
  \begin{tabular}{@{}ccc@{}}
    $\log_{10}(\text{abundance})$ & $R^2$ criterion & chosen rank $\hat{q}$ \\
    \includegraphics[width=.3\textwidth]{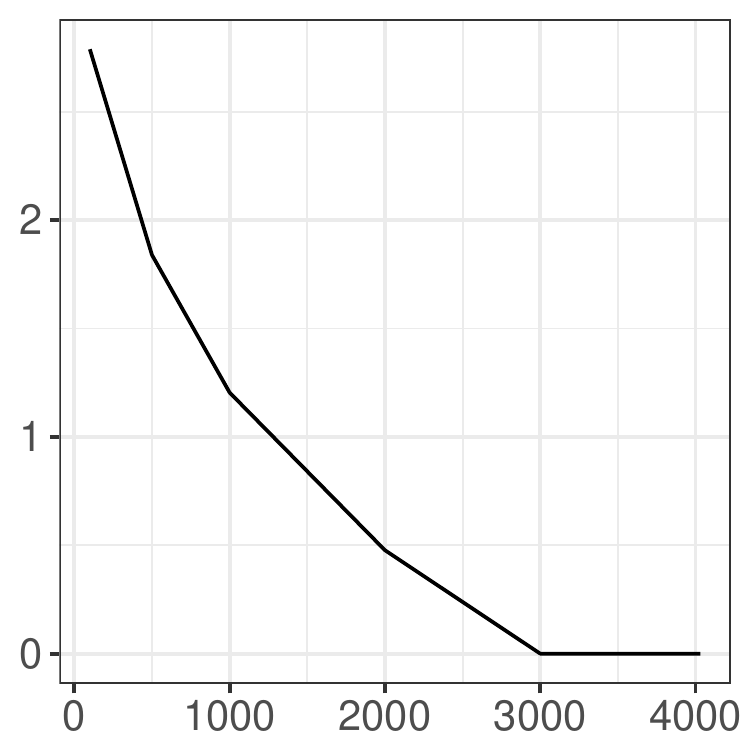} &
    \includegraphics[width=.3\textwidth]{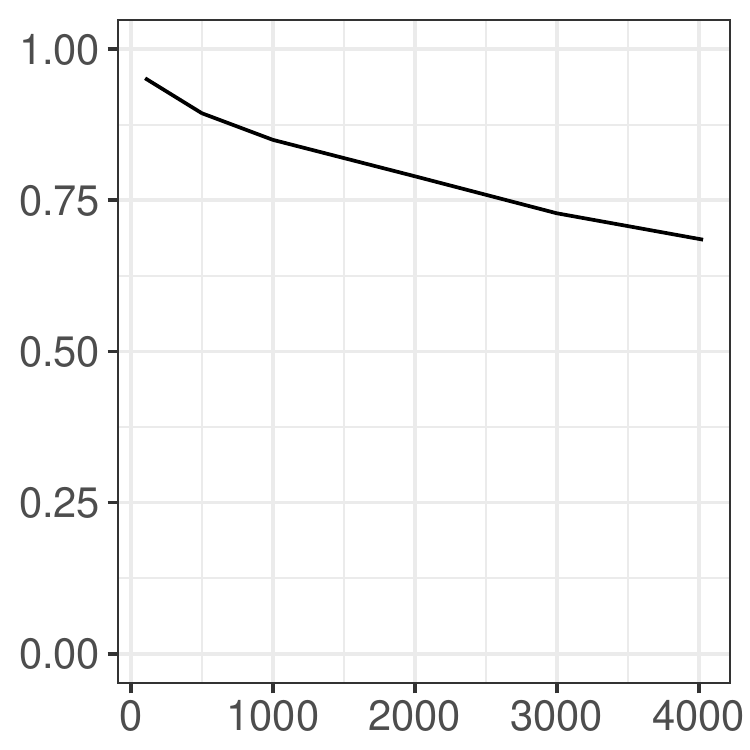} &
    \includegraphics[width=.3\textwidth]{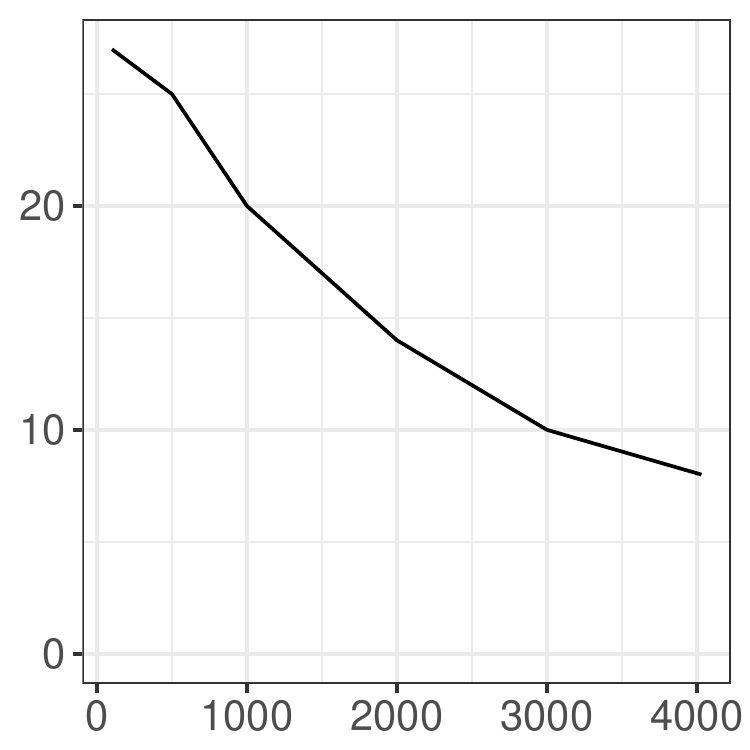} \\
     \multicolumn{3}{c}{number of variables} \\
  \end{tabular}  
  
  \caption{Dataset from \cite{Mach2015}. The minimum overall abundance
    of included OTUs (left panel), quality of approximation $R^2_q$
    (central panel) and selected value $\hat{q}$ (right panel)
    decreases when OTUs with low abundance are added to the dataset.}
  \label{Fig:Nuria1}
\end{figure}

\paragraph{Impact of Weaning}

We focus on results obtained on the $500$ most abundant OTUs, which
account for 90.3\% of the total counts. We emphasize than even doing
so, the count table remains quite sparse, with 67\% of null counts and
60\% of positive counts lower than or equal to $5$. The ICL criteria
on this subset selects $\hat{q} = 25$ ($R^2 = 89.4\%$). The main
structure present in the latent subspace is the strong and systematic
impact of weaning (Fig.~\ref{Fig:Nuria2}, left), almost entirely
captured by Axis 1. The variable factor map highlights OTUs from two
specific bacterial families: Lactobacillaceae (red) and Prevotellaceae
(blue).  The former are typically found in dairy products and thought
to be transmitted to the piglets via breast milk. As expected, they
are enriched in suckling piglets and negatively correlated with Axis
1. The latter produce enzymes that are essential to degrade cereals
introduced in the diet after weaning. As reported in \cite{Mach2015},
they are enriched after weaning and positively correlated with Axis
1. The method is thus able to recover well known structures, cope with
sparse count tables and account for varying sequencing depths.

\begin{figure}
  \centering
  \begin{tabular}{@{}c@{}}
    Individual Factor Map \hspace{6em} Variable Factor Map \\
  \includegraphics[width=.8\textwidth]{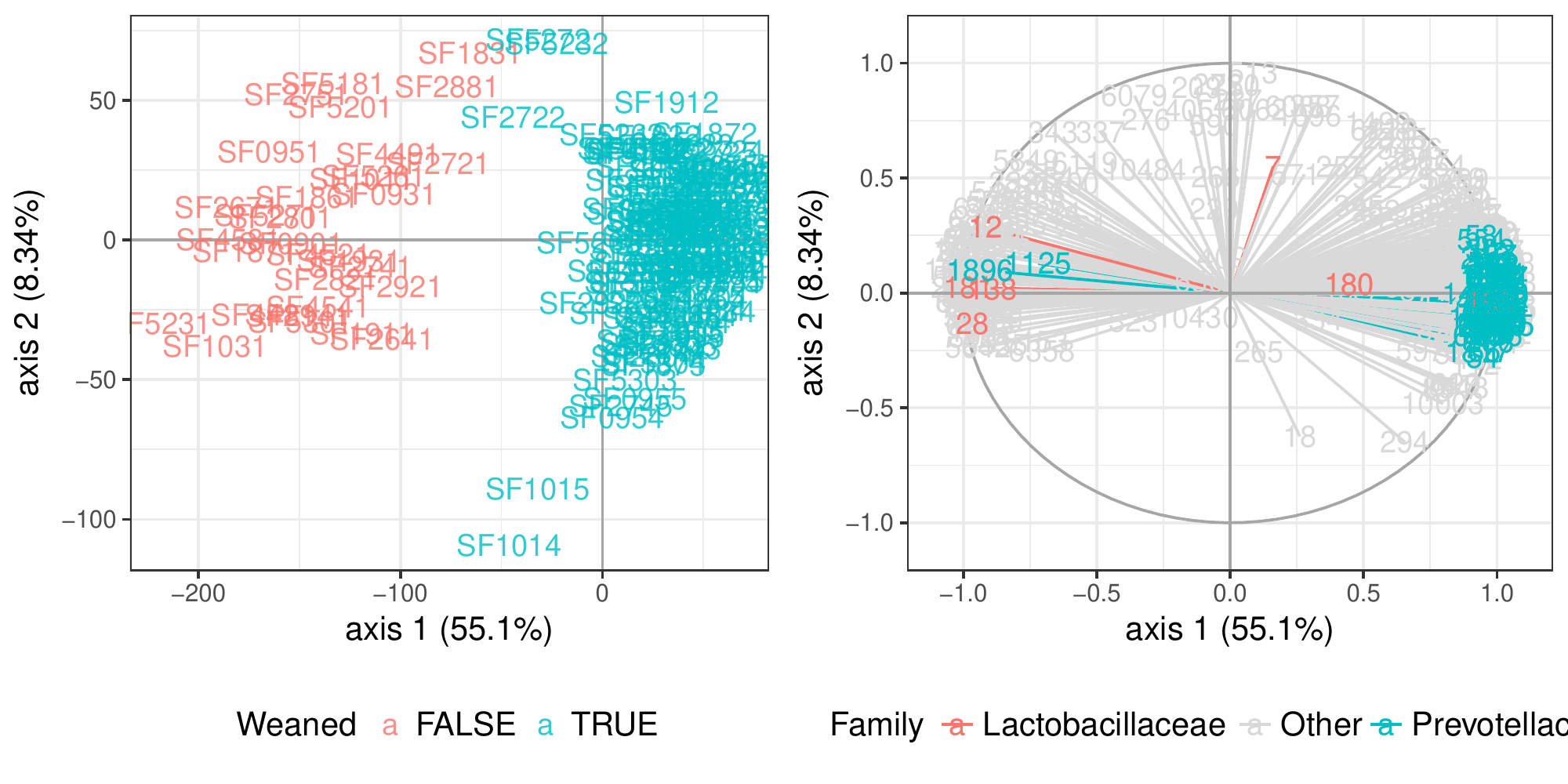}
  \end{tabular}  
  \caption{Individual (left) and variable (right) maps corresponding to the 
first principal plane of the $q$-dimensional approximation. Weaning has a strong 
and systematic effect on gut microbiota composition, well captured by axis 1. 
Bacterial families Prevotellaceae (red) and the Lactobacillaceae (blue) are two 
families well known to be affected by weaning and have a high correlation with 
Axis 1.}
  \label{Fig:Nuria2}
\end{figure}

\subsection{Oak powdery mildew pathobiome}

\paragraph{Description of the experiment}
We considered the metagenomic dataset introduced in \cite{JFS16}.
Similarly to the \cite{Mach2015} dataset, it consists of microbial
communities sampled on the surface of $n=116$ oak leaves. Communities
were sequenced with both the hypervariable V6 region of 16S rRNA as
marker-gene for bacteria and the ITS1 as marker-gene for
fungi. Sequences were cleaned, clustered at the $97$\% identity level
to create OTUs and only the most abundant ones were kept (see
\cite{JFS16} for details of OTU picking and selection) resulting in a
total of $p = 114$ OTUs (66 bacterial ones and 44 fungal ones). One
aim of this experiment is to understand the association between the
abundance of the fungal pathogenic species {\sl E. alphitoides},
responsible for the oak powdery mildew, and the other
species. Furthermore, the leaves were collected on three trees with
different resistance levels to the pathogen. In addition to the
sampling tree, several covariates, all thought to potentially
structure the community, were measured for each leaf: orientation,
distance to ground, distance to trunk, direction, etc.

We emphasize that our goal slightly differs from the one of
\cite{JFS16} as these authors were interested in reconstructing the
ecological network of the species, whereas our purpose is to summarize
the species' dependency structure in low dimension. Our approach also
differs from a methodological view-point as we jointly estimate the
effect of the covariates $\matr{\Theta}$ and the dependency structure
$\matr{\Sigma}$ while they first corrected the observed counts for the
effect of the covariates using a regression model before feeding the
residuals from that regression to a network inference method. This
two-steps procedure fails both to account for the fact that
$\matr{\Theta}$ is estimated and to propagate uncertainty from the
first step to the second one.

\begin{figure}
 \centering
 \subfloat[Model selection]{
 \label{Fig:Vacher-Criterion}
 \begin{tabular}{@{}l@{\hspace{1ex}}l@{}}
 & \hspace{3em} Offset ($M_0$) \hspace{10em} Offset and covariates ($M_1$) \\
 \rotatebox{90}{\hspace{5em} criterion value} & \includegraphics[width=.925\textwidth]{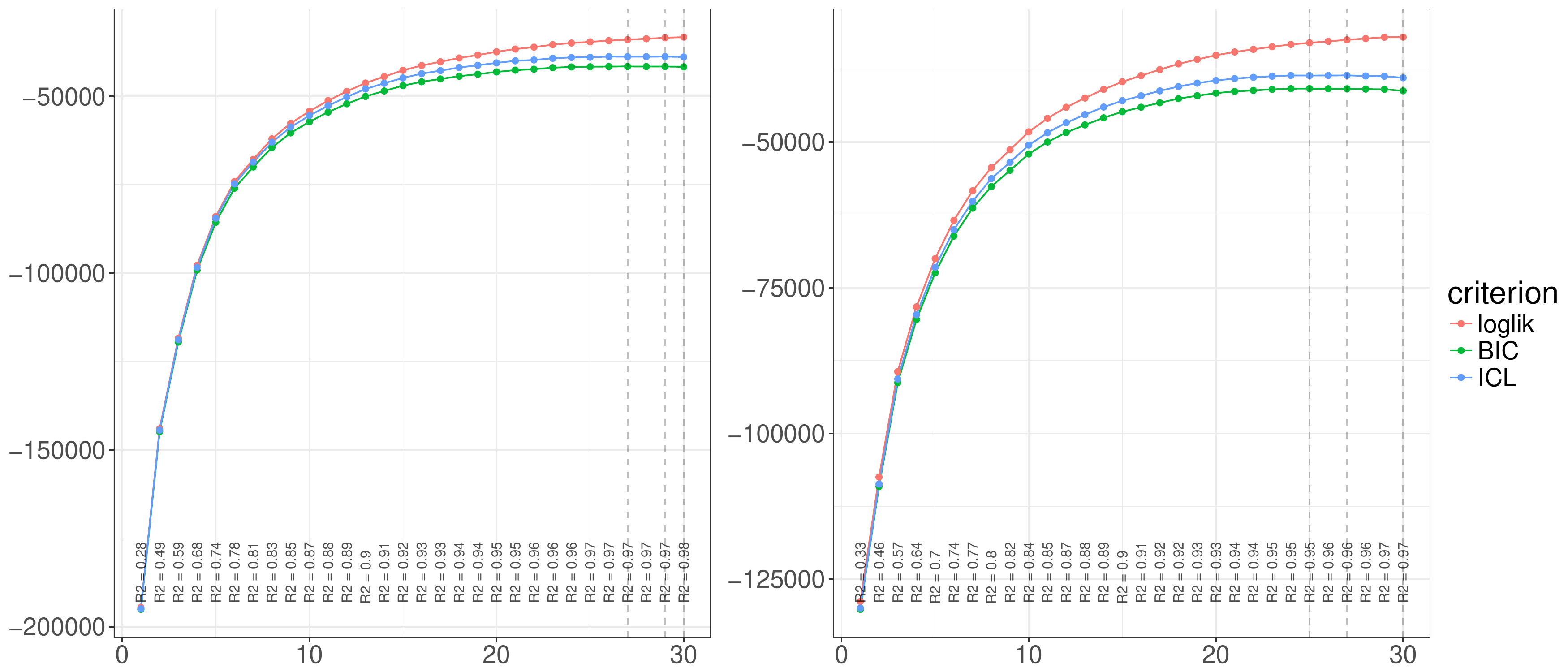} \\
 & \hspace{13em} number of axes \\
 \end{tabular} 
 }\\
 \subfloat[Goodness of fit and Entropy]{
 \label{Fig:Vacher-GoodnessEntropy} 
 \begin{tabular}{@{}l@{\hspace{5ex}}l@{}}
 & \hspace{3em} $R_q^2$ criterion \hspace{7.5em} Entropy ($\mathrm{BIC}_q - \mathrm{ICL}_q$) \\
 \rotatebox{90}{\hspace{5em} criterion value} & \includegraphics[width=.875\textwidth]{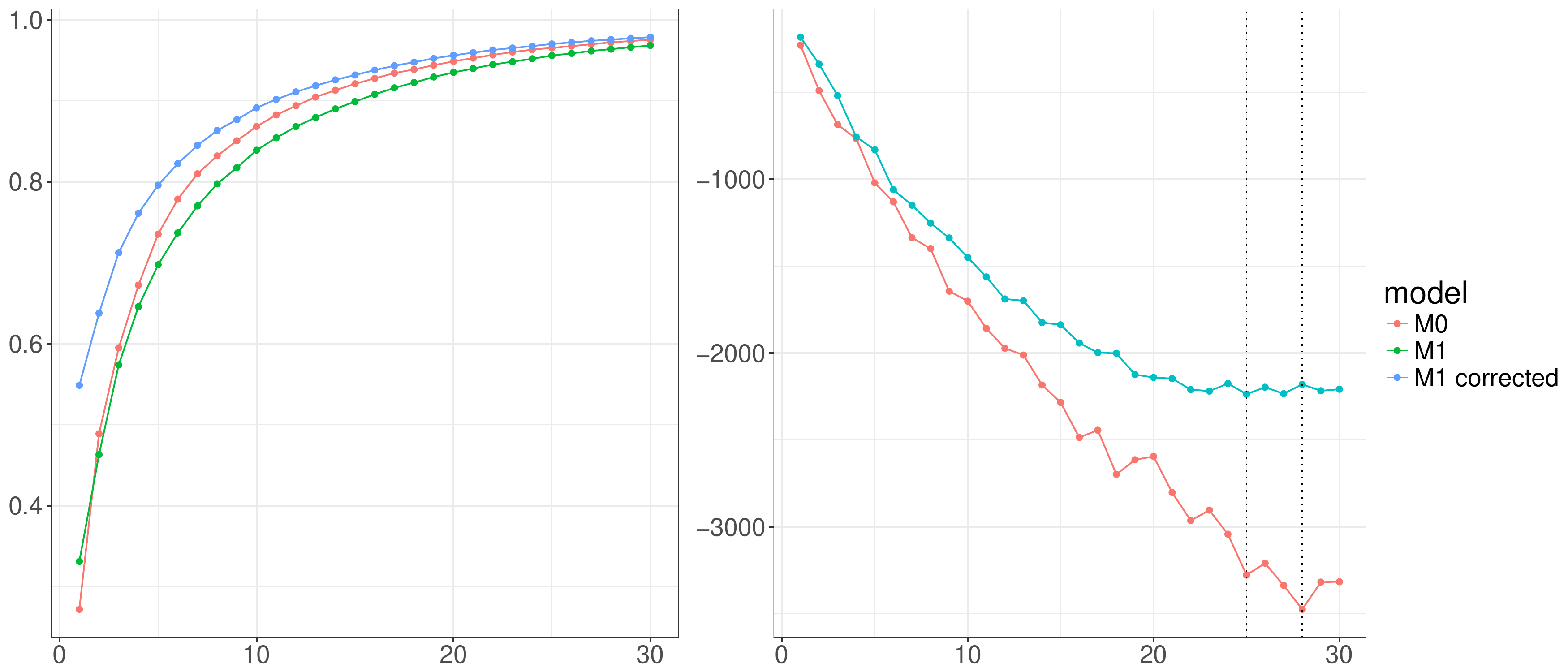} \\
 & \hspace{12em} number of axes \\
 \end{tabular} 
 }
 \caption{Dataset from \cite{JFS16}. $(a)$ model selection criteria
 $J_q$, $BIC_q$ and $ICL_q$ for model $M_0$ (left) and $M_1$
 (right); $(b)$ $R^2_q$ criterion and entropy of
 $\widetilde{p}(\matr{W})$}
 \label{Fig:Vacher}
\end{figure}

\begin{figure}
 \centering

 \subfloat[Individual Factor Maps and tree status]{
 \label{Fig:Vacher-IndMap} 
 \begin{tabular}{@{}l@{}}
 \hspace{2.5em} Offset ($M_0$) \hspace{11.5em} Offset and covariates ($M_1$) \\
 \includegraphics[width=.95\textwidth]{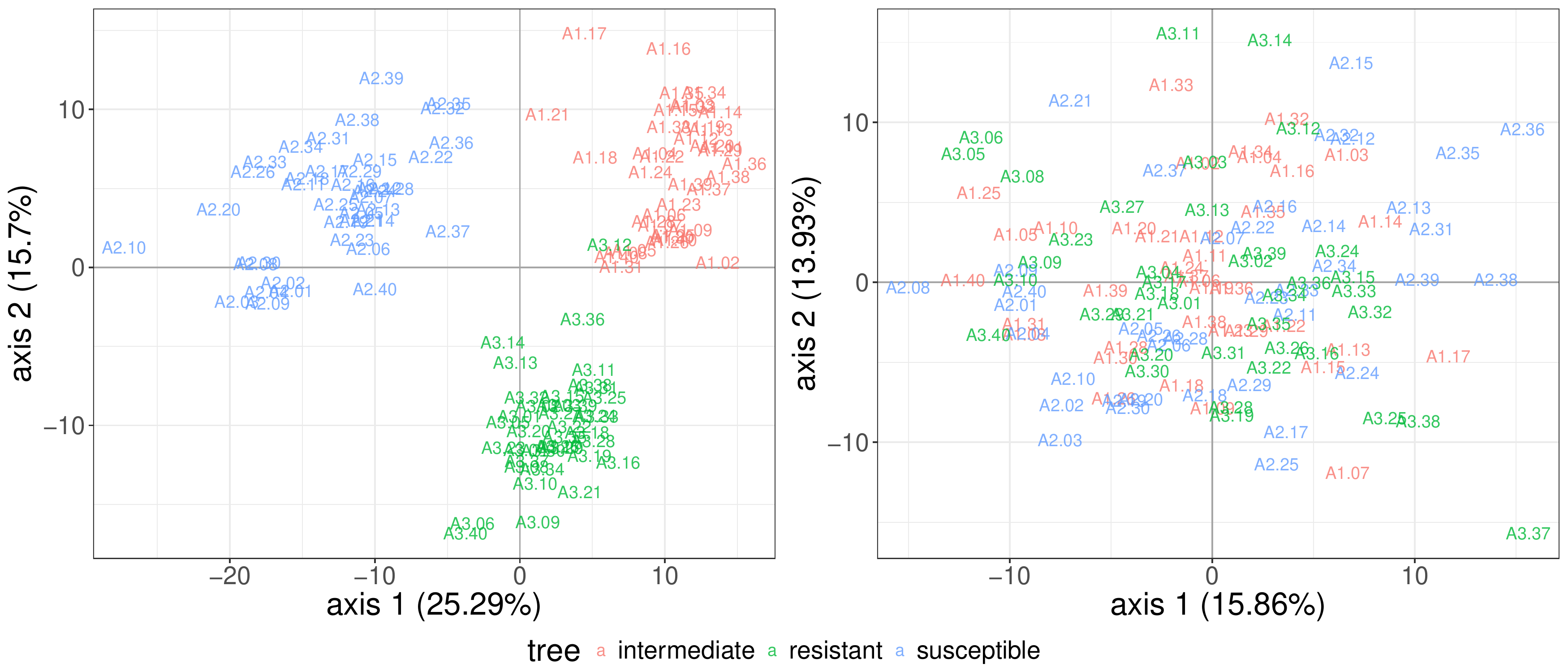} \\
 \end{tabular}
 }\\
 
 \subfloat[Individual Factor Maps and distance to ground]{
 \label{Fig:Vacher-IndMap-distTOtrunk} 
 \begin{tabular}{@{}l@{}}
 \hspace{2.5em} Offset ($M_0$) \hspace{11.5em} Offset and covariates ($M_1$) \\
 \includegraphics[width=.95\textwidth]{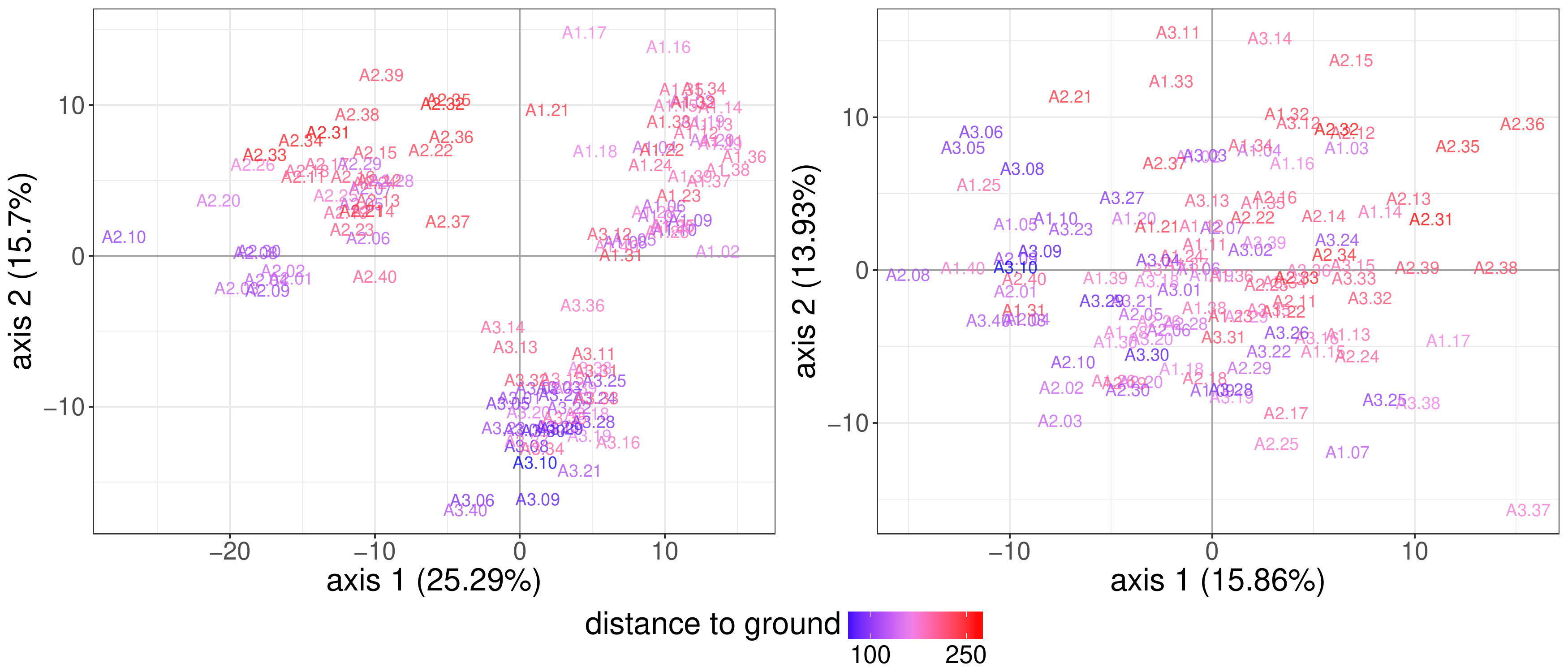} \\
 \end{tabular}
 }\\

 \caption{Dataset from \cite{JFS16}. Scatter plot of the leaves on the
   first two principal components (left: $M_0$, right: $M_1$) with
   colors corresponding to either tree status $(a)$ or distance to
   ground $(b)$. Accounting for tree status reveals an ecological
   gradient along distance to ground.}
 \label{Fig:Vacher2}
\end{figure}

\paragraph{Importance of the offset}
The abundances $Y_{ij}$ (where $i$ denotes the leaf and $j$ the OTUs)
were measured separately for fungi and bacteria resulting in different
sampling efforts for the two types of OTUs: the median total abundance
were respectively 668 for bacteria and 2166 for fungi.  To account for
this, we define a different offset $O_{ij}$ term for each OTU
type. Offsets are still computed as the log-total sums of reads,
including those of filtered out OTUs, for each OTU type.

\paragraph{Model selection}
The three trees from which the leafs where collected were respectively
susceptible, intermediately resistant (hereafter ``intermediate'') and
resistant to mildew. We first fitted a null Poisson-lognormal model
$M_0$ as defined in \eqref{eq:pca-model-with-covariates} with only an
offset term. Alternatively, we considered model $M_1$ involving two covariates:
the tree from which each leaf was collected from, and the orientation
(0=south-east, 1=north-west) of its branch.

Figure \ref{Fig:Vacher-Criterion} displays the lower bound $J$, the
BIC and the ICL for model $M_0$ (left) and $M_1$ (right) as a function
of the number of axes $q$ considered. We observe that the $J_q$ is
always increasing and that both BIC and ICL criteria behave similarly.
According to the ICL criterion, we selected $\widehat{q}_0 = 28$
($ICL = -38,619$) latent dimensions for model $M_0$ and
$\widehat{q} = 25$ ($ICL = -38,472$) for model $M_1$. This suggest that
the two models (with their respective optimal dimension) provide a
very similar fit.

We looked at the approximate posterior entropy in panel left of
Figure~\ref{Fig:Vacher-GoodnessEntropy}: we observed that it is
minimal near to the respective optimum in terms of model
selection. This indicates that the selected dimensions are also
optimal in terms of uncertainty on the latent variables.

\paragraph{Effect of the covariates}
The choice between model $M_0$ and $M_1$ is mostly a matter of the type of 
dependency we analyze with each of them, as the former does not account for the 
covariates whereas the latter does. 
This is illustrated in Figure \ref{Fig:Vacher-IndMap} (top), when plotting
the first principal plane. In model $M_0$ (left), the leafs collected
on each tree are clearly separated. As expected, taking the tree as a
covariate (right) removes the tree effect from the principal plane.

Adding covariates in the model also allows us to explore second-order
structuring effects that are masked by the strong first-order effect
of the sampling tree. Figure \ref{Fig:Vacher-IndMap} (bottom) thus
shows that in addition to sampling tree, communities are structured by
the distance of the leaf to the ground. The effect of covariates on
the abundance of {\sl E. alphitoides} were also consistent: the
estimated parameters $\theta_{ij}$ associated with the intermediate
and resistant trees were $-3.94$ and $-7.05$, respectively, taking the
susceptible tree as a reference.

We compared the respective estimates of $\matr{\Sigma}$ under $M_0$ (denoted $\widehat{\matr{\Sigma}}_0$) and under $M_1$ ($\widehat{\matr{\Sigma}}_1$) focusing on the
correlations between {\sl E. alphitoides} and the other OTUs. 
$\widehat{\matr{\Sigma}}_0$ contains correlations between OTUs that are 
either due to marginal co-variations between them or to the effect of the 
covariates, whereas the correlations in $\widehat{\matr{\Sigma}}_1$ are 
corrected from the effect of covariates. We first observed a reduction of the 
variances (median=.175, mean=.303 in $\widehat{\matr{\Sigma}}_0$; median=.087, 
mean=.176 in $\widehat{\matr{\Sigma}}_1$), which proves the strong effect 
of covariates on the abundance of the different OTUs. We then ranked 
all species according to their correlation with the pathogene and found 
very different rankings $M_0$ and $M_1$ (Kendall's $\tau =.41$), showing that 
the covariates drastically change the apparent relationship between OTU
abundances.

\paragraph{Percentage of variance} 
We now comment on use of the
$R_q^2$ criterion defined in Section \ref{sec:Visualization} to
evaluate the proportion of variability captured by a model with $q$
latent dimensions. $R_q^2$ compares the pseudo-likelihood $\ell_q^m$
obtained with $q$ latent dimensions under model $M_m$ ($m= 0, 1$) with
the likelihoods $\ell_{\min}^m$ and $\ell_{\max}^m$. We know that
$\ell_{\max}^0 = \ell_{\max}^1$ whereas $\ell_{\min}^0 <
\ell_{\min}^1$ because $\ell_{\min}^0$ only relies on the offsets
whereas $\ell_{\min}^1$ accounts for both the offsets and the
covariates. As a consequence, $R_q^2$ tends to be higher under $M_0$
than under $M_1$ for a given $q$. Right panel of Figure \ref{Fig:Vacher-GoodnessEntropy} compares the genuine $R_q^2$ under models $M_0$ and $M_1$ and the corrected version of $R_q^2$ under model $M_1$ using $\ell_{\min}^0$ in place of $\ell_{\min}^1$. As expected, the corrected version of $R_q^2$ is always higher under $M_1$ than under $M_0$. We also observe that, for both models, the proportion of variability captured by the latent space is quite high: $R_{28}^2 = 97.21\%$ for $M_0$ and $R_{25}^2 = 97.02\%$ for $M_1$. We remind that $\widehat{q}_0 = 28$ and $\widehat{q}_1 = 25$ should both be compared with $p = 114$.

\paragraph{Variance of the variational conditional distribution} 
We remind that $S_{ij}$ is the approximate conditional standard
deviation of $W_{ij}$ given the data. This parameter measures the
precision of the location of individual $i$ along the $j$-th latent
dimension. We can derive from them the approximate conditional
variance of each $Z_{ij}$ as
$[\matr{B} \diag(\vect{s}_i \odot \vect{s}_i)
\matr{B}^\intercal]_{jj}$.
Figure \ref{Corinne:PostVarZ} shows that this variance is much higher
when the corresponding abundance $Y_{ij}$ is low. Indeed, any large
negative values of $Z_{ij}$ yields a Poisson parameter close to zero
and in turn a null $Y_{ij}$. As a consequence, large negative $Z_{ij}$
cannot be predicted accurately. This is a natural consequence of the
non-linear nature of the exponential transform: large swaths of the
parameter space are compressed to small regions of the observation
space.

\begin{figure}[ht!]
 \centering
 \begin{tabular}{@{}c@{}}
 approximate conditional standard error of $Z_{ij}$ \\
 \includegraphics[width=.5\textwidth]{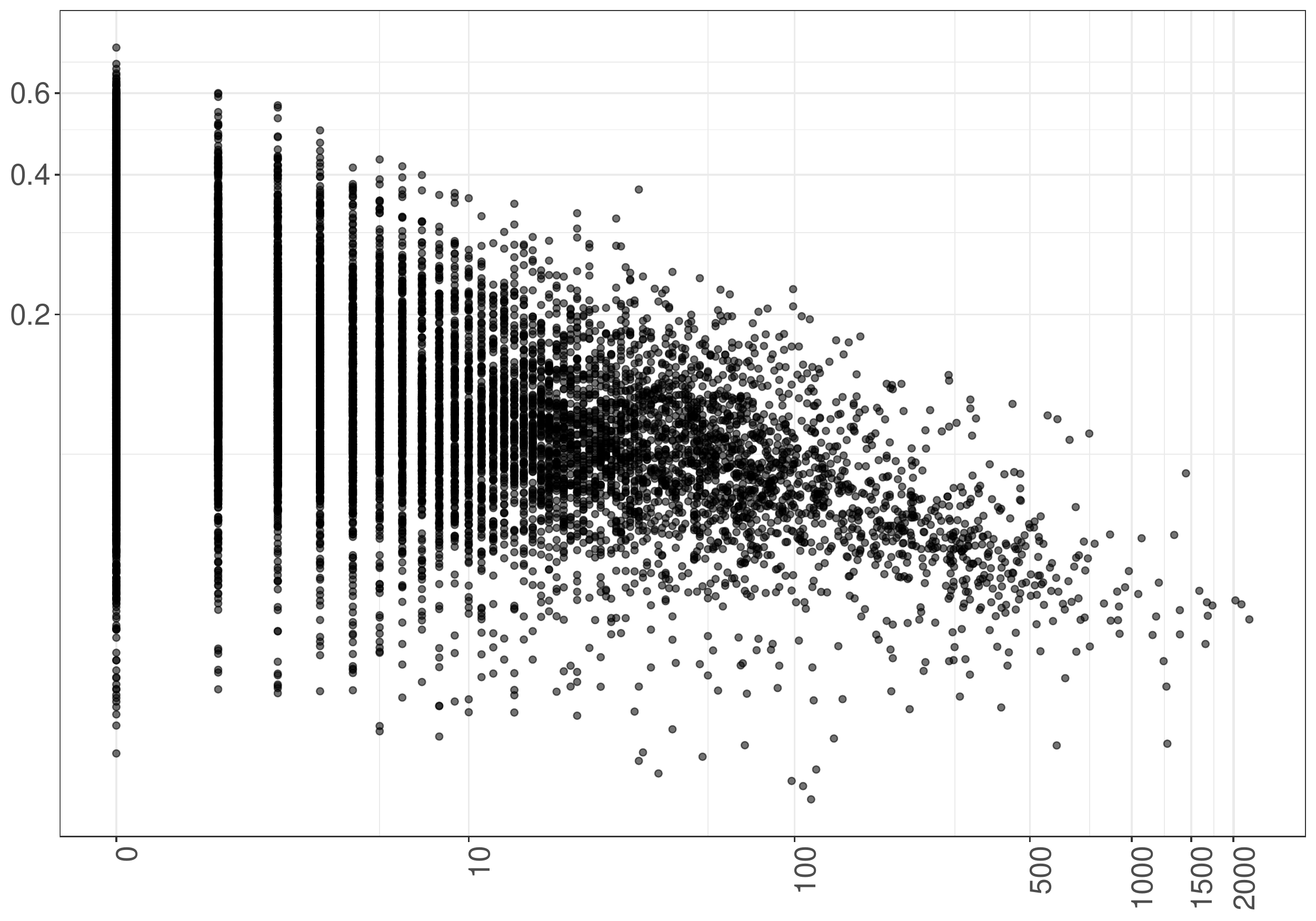}\\
 $Y_{ij}$ (log scale) \\
 \end{tabular}
 
 \caption{Variational approximate conditional standard error of the
 $Z_{ij}$ ($y$ axis) as a function of the abundance $Y_{ij}$ ($x$
 axis). \label{Corinne:PostVarZ}}
\end{figure}

\paragraph{Acknowledgement}
We thank Corinne Vacher and Nuria Mach for providing the data and
discussing the results. Two anonymous reviewers gave insightful
comments that helped improve the quality of this manuscript. We also
thank Charlie Pauvert for his feedback on the code and Elsa Teuli\`ere
for finding extra typos in the paper.

This work partially was funded by ANR Hydrogen (project
ANR-14-CE23-0001) and Inra MEM Metaprogramme (Meta-omics and microbial
ecosystems, projects LearnBioControl and Brassica-Dev).

\bibliographystyle{plainnat}
\bibliography{CMR17.bib}

\appendix
\section{Convexity Lemmas} 

\label{app:convexity}

\begin{lemma} \label{Lem:ConvExp}
For any vectors $\vect{\theta}$, $\vect{x}$, $\vect{m}$, $\vect{s}$ and $\vect{b}$ (with matching dimensions) and convex function $f$, if $\vect{u} \sim \mathcal{N}(0, \matr{I})$ and $\vect{w} = \vect{m} + \vect{s} \odot \vect{u} \sim \Ncal(\vect{m}, \diag(\vect{s} \odot \vect{s}))$, then the map $g: (\vect{\theta}, \vect{m}, \vect{s}, \vect{b}) \mapsto \Esp[ f(\crossprod{\theta}{x} + \crossprod{b}{w}) ]$ is convex in $(\vect{\theta}, \vect{b})$  for $(\vect{m}, \vect{s})$ fixed and vice-versa. 
\end{lemma}

\begin{proof}
  Note $Z  = \crossprod{\theta}{x} + \transpose{b}\vect{w} = (\crossprod{\theta}{x} + \crossprod{b}{m}) + \transpose{b} (\vect{s} \odot \vect{u})$. The first order derivative of $g$ is
  \begin{equation*}
    \vect{\nabla}(\vect{\theta}, \vect{b}, \vect{m}, \vect{s}) = \Esp \left[ f'(Z) 
      \begin{bmatrix}
        \vect{x} &  \vect{m} +  \vect{s} \odot  \vect{u} &  \vect{b} &
        \vect{b} \odot \vect{u}
      \end{bmatrix}^\trans
    \right].
  \end{equation*}
  The second order partial derivatives of $g$ are: 
  \begin{align*}
  \matr{\Psi}_1(\vect{\theta}, \vect{b}) & = \Esp \left[ f''(Z)
        \begin{bmatrix}
          \tcrossprod{x}{x} & \vect{x}(\vect{m} + \vect{u} \odot \vect{s})^\trans \\
          (\vect{m} + \vect{s} \odot \vect{u})\transpose{x} & (\vect{m} + \vect{s} \odot \vect{u})(\vect{m} + \vect{s} \odot \vect{u})^\trans \\
        \end{bmatrix}
      \right]   \\
  \matr{\Psi}_2(\vect{m}, \vect{s}) & = \Esp \left[ f''(Z)
        \begin{bmatrix}
          \tcrossprod{b}{b} & \vect{b}(\vect{b} \odot \vect{u})^\trans \\
          (\vect{b} \odot \vect{u})\transpose{b} & (\vect{b} \odot \vect{u})(\vect{b} \odot \vect{u})^\trans\\
        \end{bmatrix}  
      \right]
  \end{align*}
  And the associated quadratic form $\Phi_1(\vect{v}, \vect{w}) = (\vect{v}, \vect{w})^\trans \matr{\Psi}_1(\vect{\theta}, \vect{b}) (\vect{v}, \vect{w})$ and $\Phi_2(\vect{v}, \vect{w}) = (\vect{v}, \vect{w})^\trans \matr{\Psi}_2(\vect{m}, \vect{s}) (\vect{v}, \vect{w})$ can be simplified to 
  \begin{align*}
    \Phi_1(\vect{v}, \vect{w}) & = \Esp[f''(Z) (\crossprod{x}{v} + (\vect{m} + \vect{s} \odot \vect{u} )^\trans \vect{w})^2 ] \geq 0 \\
    \Phi_2(\vect{v}, \vect{w}) & = \Esp[f''(Z) (\crossprod{b}{v} + (\vect{b} \odot \vect{u} )^\trans\vect{w})^2 ] \geq 0
  \end{align*}
  The Hessians $\matr{\Psi}_1$ and $\matr{\Psi}_2$ are thus semidefinite positive, which ends the proof. 
\end{proof}

\begin{lemma} 
  \label{Lem:ConvExpComp}
  For any matrices $\matr{\Theta}$, $\matr{X}$, $\matr{M}$, $\matr{S}$
  and $\matr{B}$  (with matching dimensions) and  convex function $f$,
  if  $\matr{U} =  [\vect{U}_1, \dots,  \vect{U}_n]^\trans$ where  the
  $\vect{U}_i$               are               i.i.d               and
  $\vect{U}_i      \sim       \Ncal(\vect{0},      \matr{I})$      and
  $\matr{W}   =  \matr{M}   +  \matr{S}   \odot  \matr{U}$.   The  map
  $g:   (\matr{\Theta},   \matr{M},    \matr{S},   \matr{B})   \mapsto
  \transpose{1}_n        \Esp[       f(\tcrossprod{X}{\Theta}        +
  \tcrossprod{W}{B})                   ]                   \vect{1}_p$
  is convex in $(\matr{\Theta},  \matr{B})$ for $(\matr{M}, \matr{S})$
  fixed and vice-versa.
\end{lemma}

\begin{proof}
  The function $g$ is a sum of functions of the form
  $g_{ij}:  (\matr{\Theta},  \matr{M},   \matr{S},  \matr{B})  \mapsto
  \Esp[ f(\transpose{X}_i\vect{\Theta}_j + \transpose{B}_j (\vect{M}_i
  +          \vect{S}_i          \odot          \vect{U})          ]$.
  The result follows from Lemma~\ref{Lem:ConvExp}.
\end{proof}


\end{document}